\documentclass[journal]{IEEEtran}
\usepackage{amsmath,amsfonts}
\usepackage{array}
\usepackage{textcomp}
\usepackage{stfloats}
\usepackage{url}
\usepackage{verbatim}
\usepackage{graphicx}
\usepackage{cite}
\usepackage{dsfont}
\usepackage{enumerate}
\usepackage{algorithm}
\usepackage{algorithmic}
\usepackage[caption=false,font=footnotesize,labelfont=rm,textfont=rm]{subfig}
\usepackage[switch]{lineno}
\usepackage{amsmath}
\usepackage{amssymb}
\usepackage{mathtools}
\usepackage{amsthm}
\theoremstyle{plain}
\newtheorem{theorem}{Theorem}
\newtheorem{lemma}[theorem]{Lemma}

\newtheorem{corollary}[theorem]{Corollary}
\newtheorem{definition}[theorem]{Definition}

\theoremstyle{remark}
\newtheorem{remark}{Remark}
\usepackage{booktabs}

\begin{document}
	
	\title{Cooperative Game-Theoretic Credit Assignment for Multi-Agent Policy Gradients via the Core}

	\author{Mengda Ji, Genjiu Xu, Keke Jia, Zekun Duan, Yong Qiu, Jianjun Ge and Mingqiang Li
		\thanks{This work was supported by the National Key Research and Development Program of China under Grant No. 2021YFA1000402, Innovation Foundation for Doctor Dissertation of Northwestern Polytechnical University (No. CX2024090) \emph{(Corresponding author: Genjiu Xu)}}
		\thanks{Mengda Ji is with the Key Laboratory of Intelligence, Games and Information Processing in Shaanxi Province, and with the Unmanned System Research Institute, Northwestern Polytechnical University, Xi’an 710072, China. E-mail: jimengda@mail.nwpu.edu.cn.}
		\thanks{Genjiu Xu, Keke Jia, Zekun Duan and Yong Qiu are with the Key Laboratory of Intelligence, Games and Information Processing in Shaanxi Province, and with the School of Mathematics and Statistics, Northwestern Polytechnical University, Xi’an 710072, China. E-mail: xugenjiu@nwpu.edu.cn.}
		\thanks{Jianjun Ge and Mingqiang Li are with the Information Science Academy, China Electronics Technology Group Corporation, Beijing 100086, China.}
		}
		
	
	\maketitle
	\begin{abstract}
		This work focuses on the credit assignment problem in cooperative multi-agent reinforcement learning (MARL). Sharing the global advantage among agents often leads to insufficient policy optimization, as it fails to capture the coalitional contributions of different agents. In this work, we revisit the policy update process from a coalitional perspective and propose CORA, an advantage allocation method guided by a cooperative game-theoretic core allocation. By evaluating the marginal contributions of different coalitions and combining clipped double Q-learning to mitigate overestimation bias, CORA estimates coalition-wise advantages. The core formulation enforces coalition-wise lower bounds on allocated credits, so that coalitions with higher advantages receive stronger total incentives for their participating agents, enabling the global advantage to be attributed to different coalition strategies and promoting coordinated optimal behavior. To reduce computational overhead, we employ random coalition sampling to approximate the core allocation efficiently. Experiments on matrix games, differential games, and multi-agent collaboration benchmarks demonstrate that our method outperforms baselines. These findings highlight the importance of coalition-level credit assignment and cooperative games for advancing multi-agent learning.
	\end{abstract}
	
	
	\section{Introduction}
	\label{Sec:Introduction}
	
	\IEEEPARstart{C}{ooperative} Multi-Agent Reinforcement Learning (MARL) studies how a group of agents can jointly maximize a shared objective through interaction in a common environment \cite{panait2005cooperative,nguyen2020deep,zhang2023multiexperience,liu2026lrs,liu2025fuzzy,chen2025dpf}. This paradigm has shown broad potential across a range of applications \cite{hu2022marllib,wang2025bipartite,wang2025nashminmax,narayanan2025security,huang2025hitl}, including autonomous driving platoons \cite{shalev2016safe}, multi-robot systems \cite{busoniu2008comprehensive}, and large-scale network control \cite{Ma2024efficient}. A central challenge in MARL is to coordinate decentralized agents so that their local decisions collectively induce effective global strategies and maximize the team return \cite{oliehoek2008optimal,lowe2017multi,peng2026sdic}.
	
	Recent advances in policy-gradient methods have substantially improved stability and scalability in multi-agent learning. Deterministic policy-gradient methods, such as MADDPG \cite{lowe2017multi} and its variants, have shown strong performance in many continuous-control tasks. On the stochastic policy-gradient methods, MAPPO \cite{yu2021surprising}, a multi-agent extension of PPO \cite{schulman2017proximal}, has become a widely used baseline in cooperative MARL, while HAPPO and HATRPO \cite{kuba2021trust, zhong2023heterogeneousagent} further improve stability through sequential updates. However, these methods typically share the same global advantage value across agents, which can lead to suboptimal updates \cite{chai2021unmas,huang2022distributional}. The difficulty is that a shared credit signal cannot distinguish heterogeneous contributions across agents and coalitions, and may therefore provide misleading policy-update directions \cite{fu2024policy}. This motivates the development of credit assignment methods.
	
	\IEEEpubidadjcol

	To better distinguish the contributions of different agents, many studies have developed credit assignment methods. Value-based methods like VDN \cite{sunehag2017value} and QMIX \cite{rashid2018qmix}, QTRAN \cite{son2019qtran}, QPLEX \cite{wang2020qplex} and policy-gradient methods like LICA \cite{zhou2020learning}, COMA \cite{foerster2018counterfactual}, VDAC \cite{su2021value}, and FACMAC \cite{peng2021facmac} assign credit from an individual perspective and have improved coordination efficiency \cite{DBLP:conf/icml/WangZHWZGHLF22, wang2020rode}.
	Optimistic Multi-Agent Policy Gradient clips advantages and retains only positive policy updates, thereby reducing the detrimental effect of negative exploratory signals on multi-agent coordination \cite{zhao2024optimistic}.
	Among these methods, COMA is particularly representative, as it evaluates the marginal contribution of each agent to the team and assigns higher credit to agents whose actions contribute more to the global outcome.
	Relation modeling and coordinated exploration for cooperation have also been investigated, e.g., via attention and structured cognition \cite{pu2022attention, qiu2024cognition, hao2023exploration, dong2023wtoe, peng2026sdic}.
	Additionally, DOP \cite{wang2020dop} introduces value function decomposition into the multi-agent actor-critic framework, thereby supporting efficient off-policy learning and improving credit assignment under the centralized-training decentralized-execution setting.
	Overall, these methods primarily assess how individual agents contribute to team performance and then allocate agent-level incentive signals, such as values or advantages. More broadly, however, credit assignment can also be studied from a coalitional perspective.
	
	Despite their success, these methods focus exclusively on either global or individual perspectives. Between these extremes lies an underexplored middle ground: coalitional granularity, where credits are evaluated and allocated at the level of agent subsets (i.e., coalitions \( C \subseteq N \)) \cite{ding2025learning}. Cooperative game theory, traditionally used in economics for contribution attribution and payoff allocation \cite{driessen2013cooperative}, provides several classical solution concepts, including the Shapley value, the core, and the nucleolus, and has also been applied to MARL credit assignment \cite{li2025nucleolus, li2021shapley, wang2022shaq,wang2020shapley}, data valuation \cite{jia2019towards, ghorbani2019data, sim2020collaborative}, federated learning \cite{ray2022fairness, donahue2021optimality}. Recent works have introduced Shapley value-based credit assignment from cooperative game theory into policy gradient methods \cite{wang2020shapley, li2021shapley, wang2022shaq}. These approaches evaluate agents through their marginal contributions across coalitions and allocate values or advantages accordingly, so that agents with larger coalitional contributions receive higher credits. While these approaches provide grounded individual attributions, they do not directly address coalition stability. From a cooperative-game perspective, their interpretation effectively requires the induced MARL problem to be a convex cooperative game, under which the Shapley value lies in the core set \cite{li2025nucleolus, wang2020shapley,li2021shapley,wang2022shaq,driessen2013cooperative,branzei2008models}. In stochastic MARL environments, however, the induced coalitional game can be nonconvex and the exact core may even be empty. Rather than focusing on whether the exact core is nonempty, whether the Shapley value lies in the core, or how to approximate the Shapley value, we instead directly compute a relaxed core solution through a regularized $\epsilon$-core solution.

	In this paper, we propose Core Credit Assignment (CORA), a coalitional credit-assignment framework for multi-agent policy gradient methods. CORA estimates coalitional advantages by evaluating the marginal contributions of coalitions to the global return and computes a regularized least $\epsilon$-core allocation of per-agent advantages. This preserves coalitional rationality while discouraging allocations that suppress beneficial exploratory behaviors. To improve scalability, CORA employs random coalition sampling for efficient approximation.
	
	The main contributions of this paper are threefold:
	\begin{itemize}
		\item We propose a novel coalitional advantage formulation and compute a $\epsilon$-core allocation for credit assignment. The coalitions with high potential advantage values will receive higher advantage values to promote collaborative strategy optimization.
		\item We provide policy-improvement lower bounds at the coalition level, showing that the proposed method systematically reinforces beneficial coalitions.
		\item We develop a sampling approximation and demonstrate consistent performance gains across diverse MARL benchmarks, including matrix games, differential games, VMAS, SMAC, Google Research Football, and Multi-Agent MuJoCo.
	\end{itemize}
	
	
	The rest of this paper is organized as follows. Section~\ref{Sec:Background} introduces the cooperative MARL setting and discusses the limitations of shared-advantage updates. Section~\ref{sec:core_adv_decomp} presents CORA, including coalitional advantage estimation and the regularized least $\epsilon$-core advantage allocation. Section~\ref{sec:theoretical_analysis} provides theoretical analysis, including policy-improvement bounds and the sampled-coalition approximation guarantee. Section~\ref{sec:experiments} reports experimental results, and Section~\ref{sec:conclusion} concludes the paper.

	
	\section{Background}
	\label{Sec:Background}
	This section provides an overview of the foundational concepts and challenges in MARL, focusing on policy gradient methods and credit assignment methods.
	
	\subsection{Problem Formulation}
	In cooperative multi-agent reinforcement learning, a group of agents works together to maximize a shared return within a common environment \cite{panait2005cooperative, kuba2021trust}. This setting can be formalized as a Markov game \cite{littman1994markov, kuba2021trust, zhao2024optimistic} defined by the tuple $\mathcal{G} = \langle N, S, A, \mathbb{P}, r, \gamma \rangle$, where $N = \{1, \dots, n\}$ is the set of agents, $S$ is the state space, $A = \prod_{i \in N} A_i$ is the joint action space, with $A_i$ being the action space of agent $i$, $\mathbb{P}: S \times A \times S \rightarrow [0,1]$ is the transition function, $r: S \times A \rightarrow \mathbb{R}$ is the reward function, and $\gamma \in [0,1)$ is the discount factor. At each time $t \in \mathbb{N}$, each agent $i$ observes a local observation $o_i^t = O_i(s^t)$ and selects an action $a_i^t \in A_i$ drawn from its policy $\pi_i(\cdot | o_i^t;\phi_i)$. The joint action $a^t = (a_1^t, \dots, a_n^t)$ leads to the next state $s^{t+1} \sim \mathbb{P}(s^{t+1} | s^t, a^t)$ and generates a common reward $r^t = r(s^t, a^t)$ for all agents. The agents aim to update their policies to maximize the shared expected cumulative reward:
	\begin{equation}
		\max_{\pi} J(\pi) = \mathbb{E}_{s, a \sim \pi, \mathbb{P}} \left[ \sum_{t=0}^\infty \gamma^t r(s_t,a_t) \right].
	\end{equation}
	
	Under the centralized training with decentralized execution (CTDE) paradigm \cite{oliehoek2008optimal,lowe2017multi,yu2021surprising}, each agent $i$ is trained with global information and executes using only its local observation $o_i = O_i(s) \in \mathcal{O}_i$. Central components in the training process are the global state value function $V(s)$, which estimates the expected return from state $s$, and the global state-action value function $Q(s,a)$, which estimates the expected return from state $s$ after taking joint action $a$. Denoting the advantage $A(s^t,a^t) = Q(s^t,a^t) - V(s^t)$ with GAE estimator $A_{GAE}^{t} = \sum_{l=0}^{\infty} (\gamma \lambda)^l \delta_{t+l}$, where $\delta_t$ denotes the TD error $\delta_t = r_t + \gamma V(s_{t+1}) - V(s_t)$, a standard policy gradient for agent $i$ is
	\begin{equation}
		\nabla_{\phi_i}J = \mathbb{E} \left[ \nabla_{\phi_i}\log \pi_i(a_i^t | o_i^t;\phi_i) A_i(s^t, a^t) \right],
	\end{equation}
	where individual advantage $A_i(s^t, a^t)$ is the per-agent credit signal. Sharing the same advantage $A(s^t, a^t)$ fails to capture heterogeneous contributions of different agents, leading to inefficient policy updates and slower convergence.
	For concise notation, let $o=(o_1,\dots,o_n)$ with $o_i=O_i(s)$ for all $i\in N$. For any coalition $C\subseteq N$, let $o_C=(o_i)_{i\in C}$. We write $\pi_\phi(a| o)=\prod_{i\in N}\pi_i(a_i| o_i;\phi_i)$ and $\pi_C(a_C| o_C)=\prod_{i\in C}\pi_i(a_i| o_i;\phi_i)$.
	
	Throughout this paper, we focus on multi-agent credit assignment via advantage allocation for policy-gradient methods, using it to drive policy updates that strengthen effective collaboration.
	
	\subsection{Sharing Advantage}\label{sec:sharing_advantage}
	
	Many credit assignment methods such as COMA \cite{foerster2018counterfactual}, VDN \cite{sunehag2017value}, QMIX \cite{rashid2018qmix}, and LICA \cite{zhou2020learning} assign advantage or value from an individual or marginal perspective. In this paper, besides the global advantage, we consider the coalitional advantage for each coalition of agents. Let $N=\{1,\dots,n\}$ denote the set of all agents. For a given sample $(s, a)$, we evaluate the scenario where agents in coalition $C \subseteq N$ take the explored actions $a_C$, while the remaining agents follow the current policy $\pi_{N\setminus C}(\cdot | o_{N\setminus C})$.
	
	Sharing the global advantage $A(s, a)$ among agents often leads to insufficient policy updates. This approach incentivizes each agent to update its policy $\pi_i(a_i| o_i)$ to either approach action $a_i$ with $A(s, a) > 0$ or avoid those with $A(s, a) < 0$. Specifically, when an action $a$ with $Q(s, a) < V(s)$ is explored during training, all agents are penalized via $A(s, a) < 0$, and the policy $\pi_i(a_i| o_i)$ for each agent is updated to reduce its probability. This occurs even if a coalition $C$ could form a superior joint action $(a_C, a'_{N\setminus C})$ satisfying $Q(s, a_C, a'_{N\setminus C}) > V(s)$.
	
	Moreover, consider the case where the executed action $a^*$ is already optimal. If agents in coalition $C$ explore a new action $a_C$ while others act optimally, and $Q(s, a_C, a_{N\setminus C}^*) < V(s)$, then the probability $\pi_i(a_i^*| o_i)$ for each agent $i \notin C$ is reduced due to $A(s, a) < 0$, destabilizing the probability distribution over the optimal action $a^*$.
	
	In summary, the value of coalition actions can be further exploited. By evaluating the advantage of coalition $C$, each agent $i$ in a coalition with high coalitional advantage should receive stronger credit $A_i$, even when the executed joint action has a low or negative global advantage $A(s,a)$.
	
	To intuitively illustrate the limitations of sharing the global advantage, consider a two-agent matrix game scenario at state $s$, as shown in Table \ref{tab:toy_example}. In the table, $\pi_i$ denotes that agent $i$ selects actions according to its policy, while $a_i$ denotes that agent $i$ takes a deterministic action.
	
	\begin{table}[htpb]
		\centering
		\caption{Two-Agent Matrix Game: Global Advantage $A(s, a_1, a_2)$}
		\label{tab:toy_example}
		\begin{minipage}{0.45\columnwidth}
			\centering
			\caption{Example 1}
			\begin{tabular}{ccc}
				\toprule
					$A$ & $\pi_1$ & $a_1$ \\
				\midrule
				$\pi_2$ & $0$  & $-5$ \\
				$a_2$       & $-5$ & $5$ \\
				\bottomrule
			\end{tabular}
		\end{minipage}
		\hfill
		\begin{minipage}{0.45\columnwidth}
			\centering
			\caption{Example 2}
			\begin{tabular}{ccc}
				\toprule
				$A$ & $\pi_1$ & $a_1$ \\
				\midrule
				$\pi_2$ & $0$  & $5$ \\
				$a_2$       & $-5$ & $-5$ \\
				\bottomrule
			\end{tabular}
		\end{minipage}
	\end{table}
	
	In \textbf{Example 1}, the globally optimal action is $a = (a_1, a_2)$. If the agents sample this joint action, sharing the global advantage $A(s, a_1, a_2) = 5$ correctly incentivizes both agents to increase the probabilities of taking $a_1$ and $a_2$.
	
	However, in \textbf{Example 2}, the shared-advantage update can encounter a relative overgeneralization (RO) problem. Suppose the agents explore and sample the joint action $a = (a_1, a_2)$, resulting in a global advantage $A(s, a_1, a_2) = -5$. A standard shared-advantage method then assigns $A_1 = A_2 = -5$, penalizing both agents. This penalty is suboptimal for agent 1, because action $a_1$ still has high potential ($A_{\{1\}}(s, a_1)=5$). Even though the sampled global advantage is negative ($-5$), the expected coalitional advantage for agent 1 taking $a_1$ is positive ($5$).
	
	A rational credit assignment mechanism should preserve this asymmetric potential by enforcing coalition-wise lower bounds:
	$A_1 \geq A_{\{1\}}(s, a_1) - \epsilon = 5 - \epsilon$
	and
	$A_2 \geq A_{\{2\}}(s, a_2) - \epsilon = -5 - \epsilon$.
	It should also satisfy the efficiency constraint
	$A_1 + A_2 = A_{\{1,2\}}(s, a_1, a_2) = A(s, a_1, a_2) = -5$.
	Here $\epsilon \geq 0$ is a nonnegative slack variable that relaxes coalition constraints to ensure feasibility.
	For example, one feasible solution is $(A_1, A_2, \epsilon) = (2.5, -7.5, 2.5)$, which satisfies $A_1\ge 5-\epsilon$, $A_2\ge -5-\epsilon$, and $A_1+A_2=-5$. This allocation encourages using $a_1$ and suppresses the exploratory action $a_2$.

	For a three-agent illustration, let $N=\{1,2,3\}$ and consider three coalition partitions: $(\{1,2\},\{3\})$, $(\{1,3\},\{2\})$, and $(\{2,3\},\{1\})$. For each coalition $C$ given by the partitions, $\pi_C$ means policy sampling for coalition $C$, and $a_C$ means a deterministic exploratory action profile for coalition $C$.
	
	\begin{table}[!t]
		\centering
		\small
		\renewcommand{\arraystretch}{1.12}
		\setlength{\tabcolsep}{4pt}
		\caption{Coalition-level matrix examples for $N=\{1,2,3\}$}
		\label{tab:coalition_3agent}
		\begin{minipage}{0.32\columnwidth}
			\centering
			\begin{tabular}{ccc}
				\toprule
				$A$ & $\pi_{12}$ & $a_{12}$ \\
				\midrule
				$\pi_3$ & $0$  & $1$ \\
				$a_3$   & $-2$ & $-2$ \\
				\bottomrule
			\end{tabular}
		\end{minipage}
		\hfill
		\begin{minipage}{0.32\columnwidth}
			\centering
			\begin{tabular}{ccc}
				\toprule
				$A$ & $\pi_{13}$ & $a_{13}$ \\
				\midrule
				$\pi_2$ & $0$  & $0$ \\
				$a_2$   & $-2$ & $-2$ \\
				\bottomrule
			\end{tabular}
		\end{minipage}
		\hfill
		\begin{minipage}{0.32\columnwidth}
			\centering
			\begin{tabular}{ccc}
				\toprule
				$A$ & $\pi_{23}$ & $a_{23}$ \\
				\midrule
				$\pi_1$ & $0$  & $5$ \\
				$a_1$   & $-1$ & $-2$ \\
				\bottomrule
			\end{tabular}
		\end{minipage}
	\end{table}
	
	From Table~\ref{tab:coalition_3agent},
	\[
	\begin{aligned}
		A_{\{1,2\}}(s,a_{12}) &= 1,\qquad
		A_{\{1,3\}}(s,a_{13}) = 0,\\
		A_{\{2,3\}}(s,a_{23}) &= 5,\qquad
		A_{\{1,2,3\}}(s,a_1,a_2,a_3) = -2.
	\end{aligned}
	\]
	Coalitional rationality and efficiency require
	\[
	\begin{aligned}
		&A_1+A_2 \ge 1-\epsilon, \qquad & &A_1+A_3 \ge -\epsilon,\\
		&A_2+A_3 \ge 5-\epsilon, \qquad & &A_1+A_2+A_3 = -2.
	\end{aligned}
	\]
	A feasible solution is $(A_1,A_2,A_3,\epsilon)=(-4,1,1,4)$.
	This gives $A_2+A_3=2$. In contrast, the sharing advantage approach directly uses the sampled global signal
	$A(s,a_1,a_2,a_3)=-2<0$, which provides a negative update signal for coalition $\{2,3\}$.
	
	\section{Core Credit Assignment for Multi-Agent Policy Gradients}
	\label{sec:core_adv_decomp}
	In this section, we first define coalitional advantages and then present an advantage allocation algorithm based on the core solution for credit assignment.
	
	\subsection{Coalitional Advantage}\label{sec:coalitional_advantage}
	Consider a global value function $Q(s, a)$, which describes the return of the joint action $a$ in state $s$. The advantage of coalition $C$, denoted as $A_C(s, a_C)$, is defined as:
	\begin{equation}
		A_C(s, a_C) = \mathbb{E}_{a_{N \setminus C} \sim \pi_{N \setminus C}(\cdot | o_{N\setminus C})}[Q(s, a_C, a_{N \setminus C})] - V(s),
	\end{equation}
	where the first term represents the expected return when coalition $C$ takes the sampled action $a_C$, and the other agents $i \notin C$ follow the current strategy $\pi_{N \setminus C}(\cdot | o_{N\setminus C})$. Subtracting the baseline value $V(s)$ gives the advantage of coalition $C$ taking action $a_C$ alone. Incidentally, the global value naturally satisfies $A_N(s, a) = Q(s, a) - V(s) = A(s, a)$. By defining the advantage in this way, we can clearly quantify the contribution of each coalition action $a_C$ to the team.
		
	\subsection{Advantage Allocation via the Regularized Least $\epsilon$-Core}\label{sec:advantage_decomposition}
	The next problem we need to solve is how to allocate advantage $A_i(s, a)$ to each agent $i \in N$ based on $2^n$ advantage values $A_{C}(s,a_C)$ (for each $C \subseteq N$). Here $A_i(s,a)$ denotes the allocated per-agent advantage used in the actor update. Intuitively, if a coalition action $a_C$ yields a high advantage value $A_{C}(s,a_C)$, the total advantage assigned to the agents in that coalition should not be too small. Formally, we require
	\begin{equation}\label{coalitional_rationality_eps}
		\sum_{i\in C}A_i(s, a) \geq A_C (s, a_C) - \epsilon.
	\end{equation}
	This allocation principle aligns with coalitional rationality in cooperative game theory. If coalition actions $a_C$ are promising, it is beneficial to incentivize each $a_i$ ($i \in C$) to adjust its policy distribution, thereby encouraging exploration of this action in the future.
	
	Additionally, it is essential to ensure that $\sum_{i\in N} A_i(s, a) = A_N(s, a) = A(s, a)$, which is known as effectiveness in cooperative game theory and is also widely adopted as a guiding principle in value decomposition methods. For convenience, given current state $s$ and action $a$, we denote the advantage value of agent $i$, $A_i(s, a)$, simply as $A_i$.
	
	The constraints above define a strong $\epsilon$-core feasible set in cooperative game theory \cite{driessen2013cooperative}:
	\begin{equation}
		\begin{aligned}
			{} & \Bigl\{ (A_1, \cdots, A_n)\in\mathbb{R}^n \;\Bigm|\; \sum\limits_{i\in N} A_i = A_N(s,a), \sum\limits_{i\in C} A_i \ge \\
			& A_C(s,a_C) - \epsilon,\forall C\subseteq N \Bigr\}.
		\end{aligned}
	\end{equation}
	where $\epsilon \geq 0$ is a non-negative parameter that allows for a small deviation from the ideal condition.
	
	Generally, the $\epsilon$-core may admit infinitely many feasible allocations, but not all of them are desirable. In particular, some allocations satisfying coalition rationality may place all credit on a single agent, leaving others without effective incentives. We therefore compute a regularized least $\epsilon$-core allocation by minimizing $\epsilon$ together with a variance regularization term that discourages highly imbalanced solutions:
	\begin{equation}\label{eq:qp}
		\begin{aligned}
			\underset{\epsilon \geq 0, A_1, \dots, A_n}{\text{minimize}} & \quad \epsilon + \lambda_{\mathrm{reg}} \sum_{i \in N} \left(A_i - \frac{1}{|N|} A_{N} (s, a) \right)^2, \\
			\text{subject to:} 
			& \sum_{i \in N} A_i = A_{N} (s, a), \\
			& \sum_{i \in C} A_i \geq A_{C} (s, a_C) - \epsilon, \forall C \subseteq N.
		\end{aligned}
	\end{equation}
	This formulation ensures a more balanced allocation while respecting coalition rationality. Here $\lambda_{\mathrm{reg}}$ denotes the regularization coefficient. We refer to the second term in \eqref{eq:qp} as the variance regularization term; in the experimental section, we abbreviate it as the Std term to match the figure labels. In detail, $\mathbb{E}_{a_{N \setminus C} \sim \pi_{N \setminus C}(\cdot | o_{N\setminus C})}[Q(s, a_C, a_{N \setminus C})]$ can be estimated using Monte Carlo sampling, approximately given by $\frac{1}{|K|} \sum_{k \in K} Q(s, a_C, a_{N \setminus C}^k)$ where $K$ is the set of sampled trajectories, and $a_{N \setminus C}^k$ represents the action taken by the agents in $N \setminus C$ during the $k$-th trajectory.
	
	In practice, critic evaluation of unseen or exploratory coalition actions $a_C$ may suffer from overestimation bias. Such extrapolation errors can assign overly large advantages to suboptimal or risky actions, thereby distorting the core allocation. To obtain more conservative updates, we adopt clipped double Q-learning \cite{fujimoto2018addressing}: two independent state-action critics, $Q_{\theta_1}$ and $Q_{\theta_2}$, are maintained, and their minimum is used to compute the coalitional advantage:
	\begin{equation}
		A_C(s, a_C) = \min_{j \in \{1,2\}} \frac{1}{|K|} \sum_{k \in K} Q_{\theta_j}(s, a_C, a_{N \setminus C}^k) - V(s).
	\end{equation}
	This pessimistic estimate limits positive approximation errors and yields more robust coalition evaluation.
	
	Algorithm~\ref{alg:cora-v} summarizes CORA within a standard actor-critic training loop. Our implementation uses one value critic $V(s)$ to estimate the grand-coalition advantage $A_N(s,a)$ via GAE and two state-action critics $Q_{\theta_1}(s,a)$ and $Q_{\theta_2}(s,a)$ to estimate $A_C(s,a_C)$ through policy marginalization and clipped double-Q evaluation. At each update, sampled coalitions are evaluated, a constrained quadratic program produces the allocated per-agent advantages $A_i$, and these advantages guide the actor updates.

	\begin{figure}[!t]
		\centering
		\includegraphics[width=\columnwidth]{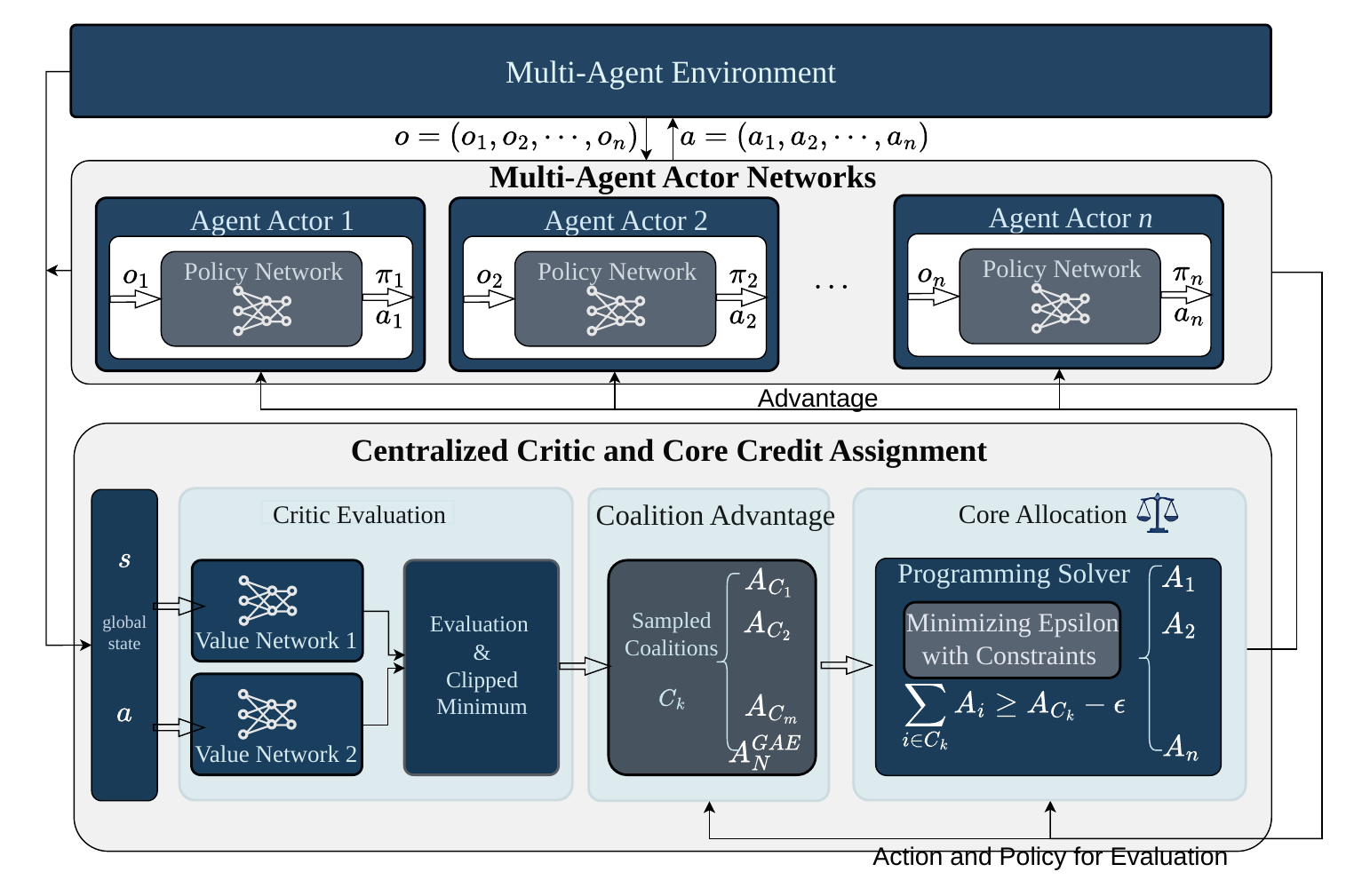}
		\caption{Overview of the CORA framework.}
	\end{figure}

	\begin{algorithm}[htbp]
		\caption{CORA: Core Advantage Allocation (PPO Policy)}
		\label{alg:cora-v}
		\begin{algorithmic}[1]
			\small
			\STATE \textbf{Initialize:} Central critic networks $\theta_V$, $\theta_{Q_1}$, $\theta_{Q_2}$; actor network $\phi_i$ for each agent $i$
			\FOR{each training episode $e = 1, \dots, E$}
			\STATE Initialize state $s^0$ and experience buffer
			\FOR{each step $t$}
			\STATE Sample action $a_i^t$ from $\pi_i(a_i^t | o_i^t;\phi_i)$ for each agent
			\STATE Execute the joint action $(a_1^t, \dots, a_n^t)$
			\STATE Get reward $r^{t}$ and next state $s^{t+1}$
			\STATE Add data to experience buffer
			\ENDFOR
			\STATE Collate episodes in buffer into a single batch
			\STATE Compute the critic target value: $y^t = r^{t} + \gamma V(s^{t+1}; \theta_V)$
			\FOR{$t = 1, \dots, T$}
			\STATE Sample $m$ coalitions $\mathcal{C} = \{ C_1, \dots, C_m \} \subseteq 2^N$
			\FOR{each coalition $C \in \mathcal{C}$}
			\STATE Estimate $Q_{C,1}(s^t, a_C^t)$ and $Q_{C,2}(s^t, a_C^t)$
			\STATE by Monte Carlo marginalization over $a_{N\setminus C}\sim\pi_{N\setminus C}(\cdot| o_{N\setminus C}^t)$
			\STATE Set $A_C(s^t, a_C^t)=\min\{Q_{C,1}(s^t, a_C^t),Q_{C,2}(s^t, a_C^t)\}-V(s^t)$
			\ENDFOR
			\STATE Estimate grand coalition advantage $A_N(s^t, a^t)$ with GAE
			\STATE Solve the programming problem to obtain allocated per-agent advantages $\hat{A}_i^t$
			\ENDFOR
			\STATE Define $r_i^t(\phi_i)=\pi_i(a_i^t | o_i^t;\phi_i)/\pi_i(a_i^t | o_i^t;\phi_i^{\mathrm{old}})$
			\STATE Update actor networks $\phi_i$ by maximizing the PPO-clipped surrogate objective:
			\[
			\begin{aligned}
			L_i^{\mathrm{PPO}}(\phi_i)=\sum_t \min\!\Big(&r_i^t(\phi_i)\hat{A}_i^t, \\
			&\mathrm{clip}\!\left(r_i^t(\phi_i),1-\epsilon_{\text{clip}},1+\epsilon_{\text{clip}}\right)\hat{A}_i^t\Big)
			\end{aligned}
			\]
			\STATE Update critic $\theta_V$ using TD error: \quad
			$\sum_t \left(V(s^t; \theta_V) - y^t\right)^2$
			\STATE Update critics $\theta_{Q_1}$ and $\theta_{Q_2}$ using errors: \quad
			$\sum_t \left(Q(s^t, a^t; \theta_{Q_1}) - y^t\right)^2$ and $\sum_t \left(Q(s^t, a^t; \theta_{Q_2}) - y^t\right)^2$
			\ENDFOR
		\end{algorithmic}
	\end{algorithm}
	
	\section{Theoretical Analysis for Core Credit Assignment}\label{sec:theoretical_analysis}
	In this section, we present theoretical results for the proposed method together with the sampling approximation used in practice. Using the standard compatible function approximation assumption in policy-gradient \cite{sutton2000policygradient,sutton2018reinforcement}, we analyze actor-critic updates (Subsections~\ref{app:fo}, \ref{app:coal-lb}, and \ref{app:concentration}). We also provide results that do not rely on this assumption, including a trust-region (TRPO/PPO-style) policy-improvement analysis and the sampled-coalition approximation guarantee (Subsections~\ref{appendix:trust-region-analysis} and \ref{app:sampled_coalitions}).
	
	\subsection{Preliminaries and Notation}\label{app:prelim}
	Let the actor parameters be $\phi=(\phi_1,\dots,\phi_n)$. The joint observation is $o=(o_1,\dots,o_n)$ with $o_i=O_i(s)$ for each $i\in N$, and the factored joint policy is
	\[
	\pi_\phi(a| o)=\prod_{i=1}^n \pi_i(a_i| o_i;\phi_i).
	\]
	When there is no ambiguity, we omit explicit dependence on $\phi_i$.
	For each agent $i$, define the score feature and Fisher matrix as
	\[
	\psi_i(o_i,a_i)=\nabla_{\phi_i}\log \pi_i(a_i| o_i;\phi_i),\qquad
	F_i=\mathbb E[\psi_i\psi_i^\top],
	\]
	and define the block-diagonal Fisher matrix
	\[
	F=\mathrm{diag}(F_1,\dots,F_n)\succeq 0.
	\]
	In practice, the empirical Fisher matrix may be singular. We therefore compute a damped inverse, i.e.,
	replace $F_i^{-1}$ by $(F_i+\beta I)^{-1}$ for a small $\beta>0$. For notational simplicity, we keep
	writing $F_i^{-1}$ and assume it is well-defined (possibly after damping).
	The NPG step is
	\begin{equation}
		\phi_i'=\phi_i+\alpha\,F_i^{-1}g_i,
		\qquad
		g_i=\mathbb E[\psi_i\,A_i],
		\label{eq:npg}
	\end{equation}
	for step size $\alpha>0$. For a sampled pair $(s,a)$, a credit allocation $\{A_i\}_{i\in N}$ satisfies the $\epsilon$-core constraints if
	\begin{equation}
		\begin{aligned}
			\sum_{i\in N}A_i &= A_N,\\
			\sum_{i\in C}A_i &\ge A_C(s,a_C) - \epsilon,
			\ \  \forall C\subseteq N.
		\end{aligned}
		\label{eq:core}
	\end{equation}

	\subsection{Compatible Function Approximation}\label{app:compatible}
	\begin{definition}[Compatible function approximation]\label{def:compatible}
		For agent $i$, consider the linear class $\mathcal S_i=\{w^\top\psi_i: w\in\mathbb R^{d_i}\}$ and let
		\[
		w_i^\star\in\arg\min_{w}\ \mathbb E\big[(A_i-w^\top\psi_i)^2\big].
		\]
		We call $A_i$ \emph{compatibly representable} if $A_i=w_i^{\star\top}\psi_i$ almost surely (equivalently, $A_i\in\mathcal S_i$). When $F_i$ is invertible, $w_i^\star$ satisfies the normal equation $\mathbb E[\psi_i A_i]=\mathbb E[\psi_i\psi_i^\top]w_i^\star=F_i w_i^\star$ (see, e.g., \cite{sutton2000policygradient,sutton2018reinforcement}).
	\end{definition}

	\begin{lemma}\label{lem:npg-w}
		Assume $F_i$ is invertible. With $g_i=\mathbb E[\psi_i A_i]$, the NPG step \eqref{eq:npg} gives
		$\phi_i'-\phi_i=\alpha\,F_i^{-1}g_i=\alpha\,w_i^\star$.
	\end{lemma}
	\begin{proof}
		From $F_i w_i^\star=g_i$ and invertibility of $F_i$, left-multiply by $F_i^{-1}$ to obtain
		$w_i^\star=F_i^{-1}g_i$.
		Substitute into \eqref{eq:npg}.
	\end{proof}

	\begin{lemma}\label{lem:firstorder}
		Assume $\log\pi_i(a_i| o_i;\phi_i)$ is twice continuously differentiable.
		Let $\Delta\phi_i=\phi_i'-\phi_i$. Then there exists $\xi_i$ on the line segment between
		$\phi_i$ and $\phi_i'$ such that
		\[
		\begin{aligned}
			\Delta\log \pi_i(a_i| o_i)
			&= \log\pi_i'(a_i| o_i)-\log\pi_i(a_i| o_i) \\
			&= \psi_i(o_i,a_i)^\top\Delta\phi_i \\
			&\quad +\frac{1}{2}\,\Delta\phi_i^\top
			\big(\nabla_{\phi_i}^2\log\pi_i(a_i| o_i;\xi_i)\big)\Delta\phi_i.
		\end{aligned}
		\]
	\end{lemma}
	\begin{proof}
		Apply the second-order Taylor theorem to $\log\pi_i(a_i| o_i;\phi_i)$ along the direction $\Delta\phi_i$.
	\end{proof}

	\subsection{First-Order Changes}\label{app:fo}
	\begin{theorem}\label{thm:main-fo}
		Assume compatible function approximation (Definition~\ref{def:compatible}) and the efficiency condition $\sum_{i\in N}A_i=A_N$.
		Consider one NPG step \eqref{eq:npg} with step size $\alpha>0$.
		Assume each $F_i$ is invertible.
		Assume for each agent $i$ that $\log\pi_i(\cdot| o_i;\phi_i)$ is twice continuously differentiable
		and its Hessian is uniformly bounded on the line segment between $\phi_i$ and $\phi_i'$:
		\begin{flalign*}
			&\big\|\nabla_{\phi_i}^2\log\pi_i(a_i| o_i;\xi_i)\big\|_{\mathrm{op}}
			\le L_i, \qquad \text{for all } \xi_i\in[\phi_i,\phi_i']. &&
		\end{flalign*}
		\begin{flalign}
			&\left|\Delta\log \pi_i(a_i| o_i)-\alpha\,A_i\right|
			\le \frac{\alpha^2}{2}\,L_i\,\big\|F_i^{-1}g_i\big\|_2^2, && \label{eq:fo-indiv}\\
			&\left|\Delta\log \pi(a| o)-\alpha\,A_N\right|
			\le \frac{\alpha^2}{2}\,\sum_{i\in N}L_i\,\big\|F_i^{-1}g_i\big\|_2^2, && \label{eq:fo-joint}\\
			&\left|\Delta\log \pi_C(a_C| o_C)-\alpha\sum_{i\in C}A_i\right|
			\le \frac{\alpha^2}{2}\,\sum_{i\in C}L_i\,\big\|F_i^{-1}g_i\big\|_2^2. && \label{eq:fo-coal}
		\end{flalign}
	\end{theorem}
	\begin{proof}
		By Lemma~\ref{lem:firstorder} and Lemma~\ref{lem:npg-w}, for each $i$,
		\[
		\begin{aligned}
			\Delta\log\pi_i(a_i| o_i)
			&=
			\psi_i^\top(\alpha F_i^{-1}g_i) \\
			&\quad +\frac{\alpha^2}{2}\,
			(F_i^{-1}g_i)^\top \\
			&\qquad \cdot
			\big(\nabla_{\phi_i}^2\log\pi_i(a_i| o_i;\xi_i)\big)(F_i^{-1}g_i).
		\end{aligned}
		\]
		Using the Hessian bound and the definition of the operator norm,
		\[
		\left|\Delta\log\pi_i(a_i| o_i)-\alpha\,\psi_i^\top F_i^{-1}g_i\right|
		\le
		\frac{\alpha^2}{2}\,L_i\,\|F_i^{-1}g_i\|_2^2.
		\]
		By Definition~\ref{def:compatible}, $\psi_i^\top F_i^{-1}g_i = A_i$, which proves \eqref{eq:fo-indiv}.
		Because $\Delta\log\pi(a| o)=\sum_i\Delta\log\pi_i(a_i| o_i)$ and $\sum_i A_i=A_N(s,a)$,
		\[
		\begin{aligned}
			\left|\Delta\log \pi(a| o)-\alpha\,A_N\right|
			&=
			\left|\sum_{i\in N}\big(\Delta\log\pi_i(a_i| o_i)-\alpha\,A_i\big)\right| \\
			&\le
			\sum_{i\in N}\left|\Delta\log\pi_i(a_i| o_i)-\alpha\,A_i\right|,
		\end{aligned}
		\]
		and applying \eqref{eq:fo-indiv} yields \eqref{eq:fo-joint}. The same argument over $i\in C$
		gives \eqref{eq:fo-coal}.
	\end{proof}

	\begin{corollary}\label{cor:prob}
		For any sampled $(s,a)$, $\pi'(a| o)=\pi(a| o)\exp(\Delta\log\pi(a| o))$, hence
		\[
		\Delta\pi(a| o)=\pi(a| o)\Big(\exp(\Delta\log\pi(a| o))-1\Big).
		\]
		Similarly, for any coalition $C$,
		$\pi_C'(a_C| o_C)=\pi_C(a_C| o_C)\exp(\Delta\log\pi_C(a_C| o_C))$ and
		\[
		\Delta\pi_C(a_C| o_C)=\pi_C(a_C| o_C)\Big(\exp(\Delta\log\pi_C(a_C| o_C))-1\Big).
		\]
	\end{corollary}

	\begin{remark}
		All first-order relations remain valid with $A_i$ replaced by its $L^2$ projection onto
		$\mathcal S_i$. Operationally, NPG realizes this via $w_i^\star=F_i^{-1}\mathbb E[\psi_i A_i]$.
	\end{remark}

	\subsection{Coalitional Lower Bounds from the Strong $\epsilon$-Core}\label{app:coal-lb}
	\begin{theorem}\label{thm:main-coal-lb-strong}
		Assume compatible function approximation (Definition~\ref{def:compatible}).
		Consider one NPG step
		$\phi_i'=\phi_i+\alpha\,F_i^{-1}g_i$ with $g_i=\mathbb{E}[\psi_i A_i]$,
		$\psi_i=\nabla_{\phi_i}\log\pi_i(a_i| o_i;\phi_i)$, $F_i=\mathbb{E}[\psi_i\psi_i^\top]$,
		and step size $\alpha>0$.
		Assume each $F_i$ is invertible.
		Assume for each agent $i$ that $\log\pi_i(\cdot| o_i;\phi_i)$ is twice continuously differentiable
		and its Hessian is uniformly bounded on the line segment between $\phi_i$ and $\phi_i'$:
		\[
		\big\|\nabla_{\phi_i}^2\log\pi_i(a_i| o_i;\xi_i)\big\|_{\mathrm{op}}\ \le\ L_i
		\quad \text{for all } \xi_i\in[\phi_i,\phi_i'].
		\]
		Then for any coalition $C\subseteq N$ and any sampled $(s,a)$,
		\begin{equation}
			\Delta\log\pi_C(a_C| o_C)
			\ \ge\
			\alpha\sum_{i\in C} A_i\;-\;
			\frac{\alpha^2}{2}\,\sum_{i\in C} L_i\,\big\|F_i^{-1}g_i\big\|_2^2.
			\label{eq:app-coal-strong}
		\end{equation}
		If, in addition, the strong $\epsilon$-core constraints hold,
		$\sum_{i\in C}A_i \ge A_C(s,a_C)-\epsilon$,
		then
		\begin{equation}
			\begin{aligned}
				\Delta\log\pi_C(a_C| o_C)
				\ \ge\ &
				\alpha\Big(A_{C}(s,a_C)-\epsilon\Big) \\
				&-
				\frac{\alpha^2}{2}\,\sum_{i\in C} L_i\,\big\|F_i^{-1}g_i\big\|_2^2.
			\end{aligned}
			\label{eq:app-coal-strong-core}
		\end{equation}
	\end{theorem}

	\begin{proof}
		For each $i$, apply the second-order Taylor expansion of
		$\log\pi_i(a_i| o_i;\phi_i)$ along the direction $\Delta\phi_i=\phi_i'-\phi_i$:
		\[
		\begin{aligned}
		\Delta\log\pi_i(a_i| o_i)
		&= \psi_i(o_i,a_i)^\top \Delta\phi_i \\
		&\quad + \frac{1}{2}\,\Delta\phi_i^\top
		\big(\nabla_{\phi_i}^2\log\pi_i(a_i| o_i;\xi_i)\big)\,\Delta\phi_i,
		\end{aligned}
		\]
		for some $\xi_i$ on the line segment between $\phi_i$ and $\phi_i'$.
		With $\Delta\phi_i=\alpha F_i^{-1}g_i$ and the operator-norm bound on the Hessian,
		\[
		\Delta\log\pi_i(a_i| o_i)
		\ \ge\
		\alpha\,\psi_i^\top F_i^{-1}g_i \;-\; \frac{\alpha^2}{2}\,L_i\,\|F_i^{-1}g_i\|_2^2.
		\]
		By Definition~\ref{def:compatible}, $\psi_i^\top F_i^{-1}g_i = A_i$.
		Summing over $i\in C$ yields \eqref{eq:app-coal-strong}.
		Combining with the strong $\epsilon$-core inequality gives \eqref{eq:app-coal-strong-core}.
	\end{proof}

	\subsection{Advantage Concentration on a Maximizing Coalition}\label{app:concentration}
	Let $C^\star\in\arg\max_{C\subseteq N}A_C(s,a_C)$. Since $N$ is included in the maximization, we have $A_{C^\star}\ge A_N$.

	\begin{theorem}\label{thm:main-concentration}
		Assume the conditions of Theorem~\ref{thm:main-fo} and the strong $\epsilon$-core constraints \eqref{eq:core}. Then:
		\begin{enumerate}
			\item $\displaystyle \sum_{i\notin C^\star}A_i
			= A_N-\sum_{i\in C^\star}A_i
			\le A_N-(A_{C^\star}-\epsilon)\le \epsilon$, hence
			\[
			\begin{aligned}
				& \Delta\log \pi_{N\setminus C^\star}(a_{N\setminus C^\star}| o_{N\setminus C^\star}) \\
				&\leq \alpha\sum_{i\notin C^\star}A_i
				+ \frac{\alpha^2}{2}\,\sum_{i\notin C^\star}L_i\,\big\|F_i^{-1}g_i\big\|_2^2 \\
				&\leq \alpha\,\epsilon
				+ \frac{\alpha^2}{2}\,\sum_{i\notin C^\star}L_i\,\big\|F_i^{-1}g_i\big\|_2^2.
			\end{aligned}
			\]
			\item By \eqref{eq:app-coal-strong-core},
			\[
			\begin{aligned}
				\Delta\log \pi_{C^\star}(a_{C^\star}| o_{C^\star})
				&\ge \alpha\big(A_{C^\star}(s,a_{C^\star})-\epsilon\big) \\
				&\quad - \frac{\alpha^2}{2}\,\sum_{i\in C^\star} L_i\,\big\|F_i^{-1}g_i\big\|_2^2.
			\end{aligned}
			\]
		\end{enumerate}
	\end{theorem}
	\begin{proof}
		From $\sum_{i\in C^\star}A_i\ge A_{C^\star}-\epsilon$ and $A_{C^\star}\ge A_N$,
		$\sum_{i\notin C^\star}A_i\le \epsilon$; then apply \eqref{eq:fo-coal} to the complement and use the
		upper bound implied by the absolute deviation.
		Apply \eqref{eq:app-coal-strong-core} to $C^\star$.
	\end{proof}

	\subsection{Trust-Region Decomposition and Lower Bounds on Policy Improvement}
	\label{appendix:trust-region-analysis}
	Under the strong $\epsilon$-core constraints \eqref{eq:core}, we derive lower bounds on policy improvement for trust-region (TRPO/PPO-style) policy updates.

	\begin{theorem}\label{thm:trpo-lower-bound}
		Given a factored joint policy $\pi_\phi(a|o)=\prod_{i\in N}\pi_i(a_i| o_i;\phi_i)$ and the CORA advantage allocation satisfying the coalition constraint
		\[
		\sum_{j\in C}A_j(s,a)\ge A_C(s,a_C)-\epsilon,\qquad \forall\,C\subseteq N,
		\]
		the following hold for the trust-region penalized policy update with parameter $\eta>0$.

		\paragraph{Individual improvement lower bound.}
		Assume $m_i\le A_i(s,a)\le M_i$ with $R_i=M_i-m_i$. Then each agent satisfies
		\[
		\Delta\log\pi_i(a_i| o_i)
		\ge
		\eta\!\left(A_i(s,a)-\epsilon\right)
		-\frac{\eta^2 R_i^2}{8}.
		\]

		\paragraph{Coalition improvement lower bound.}
		For any coalition $C\subseteq N$,
		\[
		\Delta\log\pi_C(a_C| o_C)
		\ge
		\eta\big(A_C(s,a_C)-(1+|C|)\epsilon\big)
		-
		\sum_{i\in C}\frac{\eta^2 R_i^2}{8}.
		\]
	\end{theorem}
	\begin{proof}

	Under our setting, the joint policy factorizes as $\pi_\phi(a| o)=\prod_{i\in N}\pi_i(a_i| o_i;\phi_i)$. We aim to maximize the local advantage improvement allocated to each agent while controlling the overall joint KL divergence. For clarity, we present the discrete-action case; the continuous-action case follows by replacing sums with integrals.
	
	Given action $a=(a_i,a_{-i})$, consider the global policy-update problem:
	\begin{equation}
		\begin{aligned}
			\max_{\{\pi_i'\}}\quad
			&\sum_{i\in N}\mathbb{E}_{a_i\sim\pi_i(\cdot| o_i)}\!\left[r_i(a_i| o_i)A_i(s,a)\right]\\
			\text{s.t.}\quad
			&\mathrm{KL}(\pi'(\cdot| o)\|\pi(\cdot| o))\le\delta_{\mathrm{KL}},
		\end{aligned}
	\end{equation}
	where $r_i(a_i| o_i)=\pi_i'(a_i| o_i)/\pi_i(a_i| o_i)$.
	
	Due to factorization, the joint KL satisfies
	\[
	\mathrm{KL}(\pi'\|\pi)=\sum_{i\in N}\mathrm{KL}(\pi_i'\|\pi_i),
	\]
	which implies that there is a single global trust-region constraint.
	
	Relaxing the constraint with $\eta>0$ yields the penalized form:
	\begin{equation}
		\begin{aligned}
		\max_{\{\pi_i'\}} \Bigg(
		&\sum_{i\in N}\mathbb{E}_{a_i\sim\pi_i(\cdot| o_i)}\!\left[r_i(a_i| o_i)A_i(s,a)\right]\\
		&-\frac{1}{\eta}\sum_{i\in N}\mathrm{KL}(\pi_i'(\cdot| o_i)\|\pi_i(\cdot| o_i))
		\Bigg).
		\end{aligned}
	\end{equation}
	The objective fully decomposes across $\pi_i'$, producing $n$ independent subproblems:
	\begin{equation}
		\begin{split}
			\max_{\pi_i'}\ \Big(&\mathbb{E}_{a_i\sim\pi_i(\cdot| o_i)}[r_i(a_i| o_i)A_i(s,a)]\\
			&-\tfrac{1}{\eta}\mathrm{KL}(\pi_i'(\cdot| o_i)\|\pi_i(\cdot| o_i))\Big),\qquad \forall i.
		\end{split}
	\end{equation}
	Let $q_i(a_i)=\pi_i'(a_i| o_i)$ and $p_i(a_i)=\pi_i(a_i| o_i)$.
	The subproblem becomes
	\begin{equation}
		\max_{q_i}\quad
		\sum_{a_i}q_i(a_i)A_i(s,a) - \frac{1}{\eta}\sum_{a_i}q_i(a_i)\log\frac{q_i(a_i)}{p_i(a_i)},
	\end{equation}
	subject to $\sum_{a_i}q_i(a_i)=1$.
	
	Construct the Lagrangian
	\[
	\begin{aligned}
	\mathcal{L}(q_i,\lambda)
	&=\sum_{a_i}q_i(a_i)A_i(s,a) - \frac{1}{\eta}\sum_{a_i}q_i(a_i)\log(q_i/p_i)\\
	&\quad+\lambda\!\left(\sum_{a_i}q_i-1\right).
	\end{aligned}
	\]
	Taking the derivative w.r.t.\ $q_i(a_i)$ and setting it to zero yields
	$\log(q_i/p_i)=\eta A_i(s,a)+c$.
	Thus the optimal update is
	\begin{equation}
		\begin{aligned}
		\pi_i'(a_i| o_i)
		&=\frac{\pi_i(a_i| o_i)\exp(\eta A_i(s,a))}{Z_i(s,a_{-i})},\\
		Z_i
		&=\mathbb{E}_{a_i\sim\pi_i(\cdot| o_i)}[\exp(\eta A_i(s,a))],
		\end{aligned}
	\end{equation}
	and consequently
	\[
	\Delta\log\pi_i(a_i| o_i)
	= \eta A_i(s,a)-\log Z_i(s,a_{-i}).
	\]
	If $m_i\le A_i(s,a)\le M_i$ (define $R_i=M_i-m_i$), let $X=A_i(s,a)$ with $a_i\sim\pi_i(\cdot| o_i)$.
	Hoeffding's lemma gives
	\[
	\log Z_i\le\eta\mathbb{E}[X]+\frac{\eta^2R_i^2}{8}.
	\]
	Hence,
	\begin{equation}
		\begin{aligned}
			\Delta\log\pi_i(a_i| o_i)
			&\ge \eta\big(A_i(s,a)-\mathbb{E}_{a_i\sim\pi_i(\cdot| o_i)}[A_i(s,a)]\big)\\
			&\quad-\frac{\eta^2R_i^2}{8}.
		\end{aligned}
	\end{equation}
	The CORA coalition constraint states that for any $C\subseteq N$,
	\[
	\sum_{j\in C}A_j(s,a)\ge A_C(s,a_C)-\epsilon.
	\]
	It can be proven that $\mathbb{E}_{a_i \sim \pi_i(\cdot| o_i)} [A_i(s, a)] \leq \epsilon$. For any given $(s, a)$, we have $A_i(s, a) = A_N(s, a) - \sum_{j \neq i} A_j(s, a)$, and therefore
	\[
	A_i(s, a) \leq A_N(s, a) - A_{N-i}(s, a_{N-i}) + \epsilon.
	\]
	Since $A_N = Q(s, a) - V(s)$ and $A_{N-i}(s,a_{N-i}) = Q_{N-i}(s,a_{N-i}) - V(s)$ with $Q_{N-i}(s,a_{N-i})=\mathbb{E}_{a_i \sim \pi_i(\cdot| o_i)}[Q(s, a)]$, we obtain
	\[
	\begin{aligned}
		\mathbb{E}_{a_i \sim \pi_i(\cdot| o_i)}[A_i(s, a)]
		\leq\ &
		\mathbb{E}_{a_i \sim \pi_i(\cdot| o_i)}[Q(s, a) - V(s)] \\
		&- \big(Q_{N-i}(s, a_{N-i}) - V(s)\big) + \epsilon.
	\end{aligned}
	\]
	Since $\mathbb{E}_{a_i \sim \pi_i(\cdot| o_i)}[Q(s, a)] = Q_{N-i}(s, a_{N-i})$, it follows that $\mathbb{E}_{a_i \sim \pi_i(\cdot| o_i)}[A_i(s, a)] \leq \epsilon$.
	
	Thus the final lower bound on the individual log-probability improvement becomes
	\begin{equation}
		\Delta\log\pi_i(a_i| o_i)
		\ge
		\eta(A_i(s,a)-\epsilon)
		-\frac{\eta^2R_i^2}{8}.
	\end{equation}
	The coalition log-probability change is
	$\Delta\log\pi_C(a_C| o_C)=\sum_{i\in C}\Delta\log\pi_i(a_i| o_i)$.
	Thus,
	\begin{equation}
		\Delta\log\pi_C(a_C| o_C)
		\ge
		\eta\Big(\sum_{i\in C}A_i(s,a)-|C|\epsilon\Big)
		-
		\sum_{i\in C}\frac{\eta^2R_i^2}{8}.
	\end{equation}
	Applying the coalition advantage constraint
	$\sum_{i\in C}A_i\ge A_C(s,a_C)-\epsilon$
	gives a tighter lower bound:
	\begin{equation}
		\Delta\log\pi_C(a_C| o_C)
		\ge
		\eta\big(A_C(s,a_C)-(1+|C|)\epsilon\big)
		-
		\sum_{i\in C}\frac{\eta^2R_i^2}{8}.
	\end{equation}
	
	\end{proof}

	These bounds indicate that, in addition to increasing the probability of the global action $a_N$ by $A_N$, CORA promotes action profiles $(a_C,\pi_{N\setminus C})$ when they exhibit high coalitional advantages $A_C$.
	
	Next, consider two representative cases. If $A_C(s,a_C)\ge A_N(s,a)\ge 0$, then
	$\sum_{i\in C}A_i \ge A_C(s,a_C)-\epsilon \ge A_N(s,a)-\epsilon$,
	and therefore
	\[
	\sum_{i\notin C}A_i(s,a)
	= A_N(s,a)-\sum_{i\in C}A_i(s,a)
	\le \epsilon.
	\]
	Hence, a high-value coalition $C$ receives almost all of the available advantage, up to the relaxation term $\epsilon$.
	
	If instead $A_C(s,a_C)\ge 0 > A_N(s,a)$, then
	$\sum_{i\in C}A_i \ge A_C(s,a_C)-\epsilon \ge -\epsilon$,
	and
	\[
	\sum_{i\notin C}A_i(s,a)
	= A_N(s,a)-\sum_{i\in C}A_i(s,a)
	\le A_N(s,a)+\epsilon.
	\]
	In this case, minimizing $\epsilon$ limits the penalty assigned to coalition $C$, while the larger negative credit is absorbed by $N\setminus C$. Consequently, even when the sampled global action is poor, a coalition with positive coalitional value can still preserve its action probability. This mechanism is a key principle underlying CORA.

	\subsection{Approximation with Sampled Coalitions}\label{app:sampled_coalitions}
	The approximate quadratic programming problem mentioned in the main text is as follows.

	\begin{align}\label{eq:approx_qp}
		\underset{\epsilon \geq 0, A_1, \dots, A_n}{\text{minimize}} & \quad \epsilon + \lambda_{\mathrm{reg}} \sum_{i \in N} \left(A_i - \frac{1}{|N|} A_N \right)^2, \notag \\
		\text{subject to:} 
		& \sum_{i \in N} A_i =  A_N, \notag \\
		& \sum_{i \in C_k} A_i \geq A_{C_k}(s, a_{C_k}) - \epsilon, \forall C_k \in \mathcal{C}.
	\end{align}
	The proof of Theorem \ref{thm_coalition_approx} follows an approach inspired by \cite{yan2021if}, where the core allocation is approximated using sampled coalitions. The key idea is to leverage the properties of the VC-dimension of a function class to bound the probability of deviating from the true allocation in the core. To establish this result, we introduce the following two known lemmas, which play a crucial role in the proof.

	Before proving the theorem, we first introduce a lemma regarding the VC-dimension of a function class, as this concept is essential to understanding the behavior of the classifier we employ in the proof.

	\begin{lemma}\label{appendix_lemma1}
		Let \( \mathcal{F} \) be a function class from \( \mathcal{X} \) to \( \{-1, 1\} \) with VC-dimension \( d \), and let \( \text{y}:\mathcal{X}\to\{-1,1\} \) be a target labeling function. Then, with \( m = O \Big( \frac{d \log\left( \frac{1}{\delta} \right) + \log\left( \frac{1}{\Delta} \right)}{\delta^2} \Big) \) i.i.d. samples \( \{ \text{x}^1, \dots, \text{x}^m \} \sim \mathcal{P} \), we have:
		\[
		\Big| \mathop{\text{Pr}}_{x \sim \mathcal{P}} [f(\text{x}) \neq \text{y}(\text{x})] - \frac{1}{m}\sum_{i = 1}^{m} \mathbb{1}_{f(\text{x}^i) \neq \text{y}(\text{x}^i)} \Big| \leq \delta,
		\]
		for all \( f \in \mathcal{F} \) and with probability \( 1 - \Delta \).
	\end{lemma}

	This lemma states that if the VC-dimension of a function class is \( d \), then with a sufficient number of samples \( m \), the empirical error rate of a classifier \( f \) is close to its true error rate with probability at least \( 1 - \Delta \).

	For Theorem \ref{thm_coalition_approx}, we use linear classifiers to represent the sampled strong $\epsilon$-core feasibility constraints. The following lemma gives the VC-dimension of the resulting class of linear classifiers.

	\begin{lemma}\label{appendix_lemma2}
		For $d\in\mathbb N$, the class of homogeneous linear classifiers
		\[
		\mathcal{F}^{d}=\left\{\text{x}\mapsto \text{sign}(\text{w}\cdot \text{x}): \text{w}\in\mathbb R^{d}\right\}
		\]
		has VC-dimension $d$.
	\end{lemma}

	This lemma shows that the VC-dimension equals the ambient dimension $d$; we apply it with $d=n+2$.

	\begin{theorem}\label{thm_coalition_approx}
		Given a distribution $\mathcal{P}$ over $2^N$ and parameters $\delta,\Delta>0$, solving (\ref{eq:approx_qp}) over $O(((n+2)\log(1/\delta)+\log(1/\Delta))/\delta^2)$ coalitions sampled from $\mathcal{P}$ yields a sampled regularized least $\epsilon$-core allocation that lies in the $\delta$-probable core with probability at least $1 - \Delta$.
	\end{theorem}
	\begin{proof}
		Consider a coalition \( C \) sampled from the distribution \( \mathcal{P} \). We represent the coalition as a vector \( \text{z}^C = (I^C, -A_{C}(s,a_C), 1) \), where \( I^C \in \{0,1\}^n \) is the indicator vector for the coalition and \( A_{C}(s,a_C) \) is the coalitional advantage (value) of coalition \( C \).
		
		We define a linear classifier \( f \) with parameters \( \text{w}^f = (A, 1, \epsilon) \in \mathbb{R}^{n+2} \). The classifier \( f(\text{z}^C) = \text{sign}(\text{w}^f \cdot \text{z}^C) \) is designed to encode whether the sampled strong $\epsilon$-core constraint is satisfied for coalition \( C \).
		
		To ensure coalition rationality, we require the classifier \( f \) to satisfy \( f(\text{z}^C) = 1 \) for all coalitions \( C \subseteq N \). This means that the allocation satisfies the corresponding feasibility constraints for all coalitions. The class of such classifiers is:
		
		\[
		\mathcal{F} = \left\{ \text{z} \mapsto \text{sign}(\text{w} \cdot \text{z}) : \text{w} = (A,1,\epsilon),\ A \in \mathbb{R}^n,\ \epsilon \in \mathbb{R} \right\}.
		\]
		This class of functions \( \mathcal{F} \) has VC-dimension at most \( n+2 \) by Lemma \ref{appendix_lemma2}.
		
		Now, solving the quadratic programming problem on \( m \) samples of coalitions \( \{C_1, \cdots, C_m\} \) provides a solution \( (\hat{A}, \hat{\epsilon}) \), and the corresponding classifier \( \hat{f} \). For each sample coalition \( C_k \), we have \( \hat{f}(\text{z}^{C_k}) = 1 \).
		
		By applying Lemma \ref{appendix_lemma1}, with probability \( 1 - \Delta \), we obtain the following inequality:
		
		\[
		\mathop{\text{Pr}}\limits_{C \sim \mathcal{P}} \Big[\sum_{i \in C} \hat{A}_i - A_{C}(s, a_C) + \hat{\epsilon} \geq 0\Big] \geq 1 - \delta.
		\]
		Hence, the sampled regularized least $\epsilon$-core allocation obtained from the approximate quadratic program lies in the \( \delta \)-probable core with probability at least \( 1 - \Delta \).
	\end{proof}

	\section{Experiments}
	\label{sec:experiments}
	
	We evaluate the CORA method across several cooperative multi-agent environments, including matrix games, differential games, the VMAS simulator \cite{bettini2022vmas}, the Multi-Agent MuJoCo (MaMujoco) environment \cite{gymnasium_robotics2023github, kuba2021trust}, and the Starcraft Multi-Agent Challenge (SMAC) environment \cite{samvelyan2019starcraft, hu2022marllib}.

	All experiments were conducted on a workstation equipped with an AMD 7970X 32-Core CPU, 128GB RAM, and an NVIDIA RTX 4090 GPU (24GB). Unless otherwise specified, all algorithms employed a two-layer multilayer perceptron (MLP) with hidden width 64. Each configuration was evaluated over 5 or 8 independent runs with different random seeds for both the algorithm and the environment. Unless otherwise noted, credit assignment was performed using a subset containing approximately half of all possible coalitions, specifically $2^{n-1}-1$ coalitions. For computational efficiency, 64 parallel environments were used for the \textit{Navigation} task in VMAS, while the remaining settings used 16 or 4 parallel environments, as summarized in Table~\ref{tab:exp_settings}.

	\begin{table*}[!t]
		\centering
		\caption{Training hyperparameters across environments.}
		\label{tab:exp_settings}
		\scriptsize
		\setlength{\tabcolsep}{4.2pt}
		\begin{tabular}{lccccccc}
			\toprule
			\textbf{Environment} & \textbf{Actor LR} & \textbf{Critic LR} & \textbf{Clip} $\epsilon_{\text{clip}}$ & \textbf{Entropy Coef.} & \textbf{MC Samples} $K$ & \textbf{Parallel Envs} & \textbf{Runs} \\
			\midrule
			Matrix Games         & $5 \times 10^{-4}$ & $5 \times 10^{-3}$ & 0.3 & $10^{-3}$ & 16 & 4  & 5 \\
			Differential Games   & $5 \times 10^{-5}$ & $5 \times 10^{-4}$ & 0.2 & $10^{-4}$ & 16 & 4  & 5 \\
			Multi-Agent MuJoCo   & $5 \times 10^{-4}$ & $5 \times 10^{-3}$ & 0.2 & $10^{-4}$ & 64 & 4  & 5 \\
			VMAS (Navigation)    & $5 \times 10^{-4}$ & $5 \times 10^{-3}$ & 0.2 & $10^{-4}$ & 16 & 64 & 5 \\
			VMAS (Others)        & $5 \times 10^{-4}$ & $5 \times 10^{-3}$ & 0.2 & $10^{-4}$ & 16 & 16 & 5 \\
			SMAC                 & $5 \times 10^{-4}$ & $5 \times 10^{-4}$ & 0.2 & $10^{-3}$ & 64 & 20 & 8 \\
			GRF                  & $5 \times 10^{-4}$ & $5 \times 10^{-4}$ & 0.2 & $10^{-3}$ & 64 & 50 & 8 \\
			\bottomrule
		\end{tabular}
		\vspace{2pt}
		
		\footnotesize Shared across all environments: $\gamma=0.99$, GAE $\lambda=0.95$, epochs/update $=10$, and core regularization $\lambda_{\mathrm{reg}}=10^{-2}$.
	\end{table*}
	
	\subsection{Matrix Games}
	
	In this section, we construct two matrix-style cooperative game environments to evaluate the fundamental performance of different algorithms.
	
	\textbf{Matrix Team Game (MTG)}: In this environment, agents receive a shared reward at each time step based on their joint action, determined by a randomly generated reward matrix. Each element of the matrix is uniformly sampled from the interval $[-10, 20]$. The game proceeds for 10 steps per episode. Each agent observes a global one-hot encoded state indicating the current step number, allowing them to learn time-dependent coordination strategies.
	
	\textbf{Multi-Peak Matrix Team Game}: To further evaluate each algorithm's ability to optimize cooperative strategies in environments with multiple local optima, we extend MTG to design a more challenging setting. The matrix is filled with background noise in the range $[-10, 0]$, overlaid with multiple reward peaks. Among them, one peak is the global optimum (highest value), while the rest are local optima. Actions deviating from peak combinations incur heavy penalties due to the negative background. This setting is designed to test whether algorithms can escape suboptimal solutions and discover globally coordinated strategies.
	
	\begin{figure}[t]
		\centering
		\subfloat[Base (Matrix Team Game)]{%
			\includegraphics[width=0.48\columnwidth]{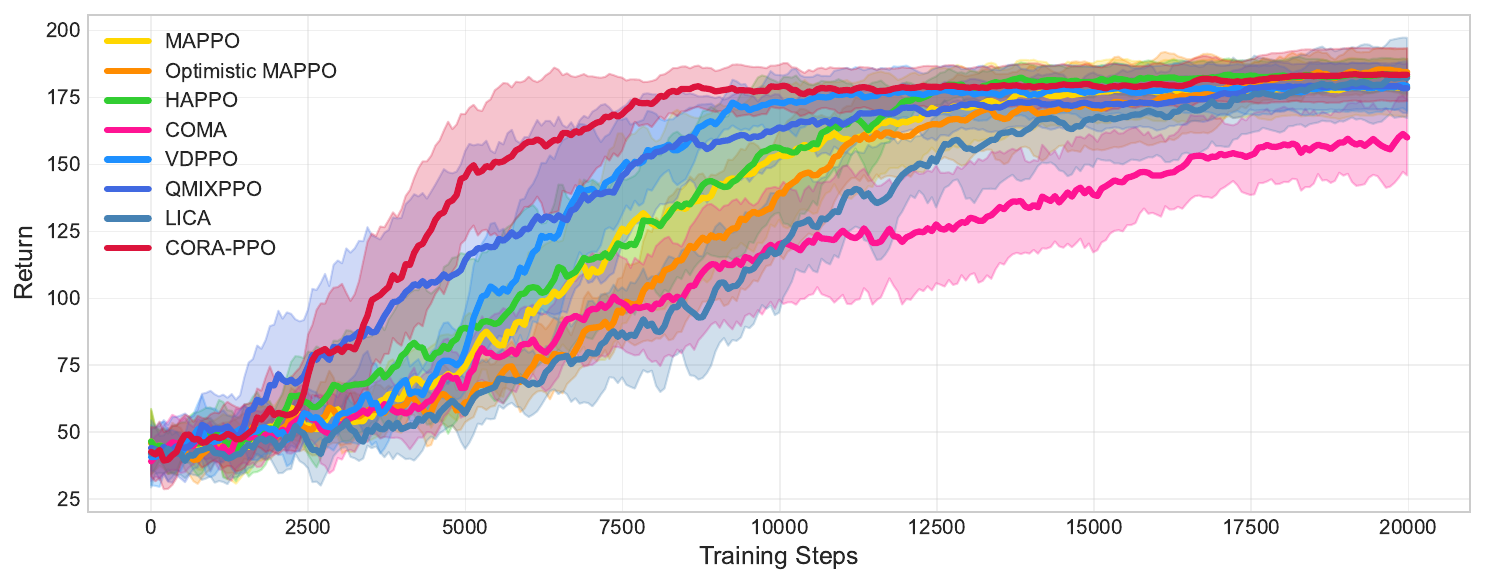}
			\label{fig:matrix_team_game_base}
		}
		\hfill
		\subfloat[5 Peaks (Matrix Team Game)]{%
			\includegraphics[width=0.48\columnwidth]{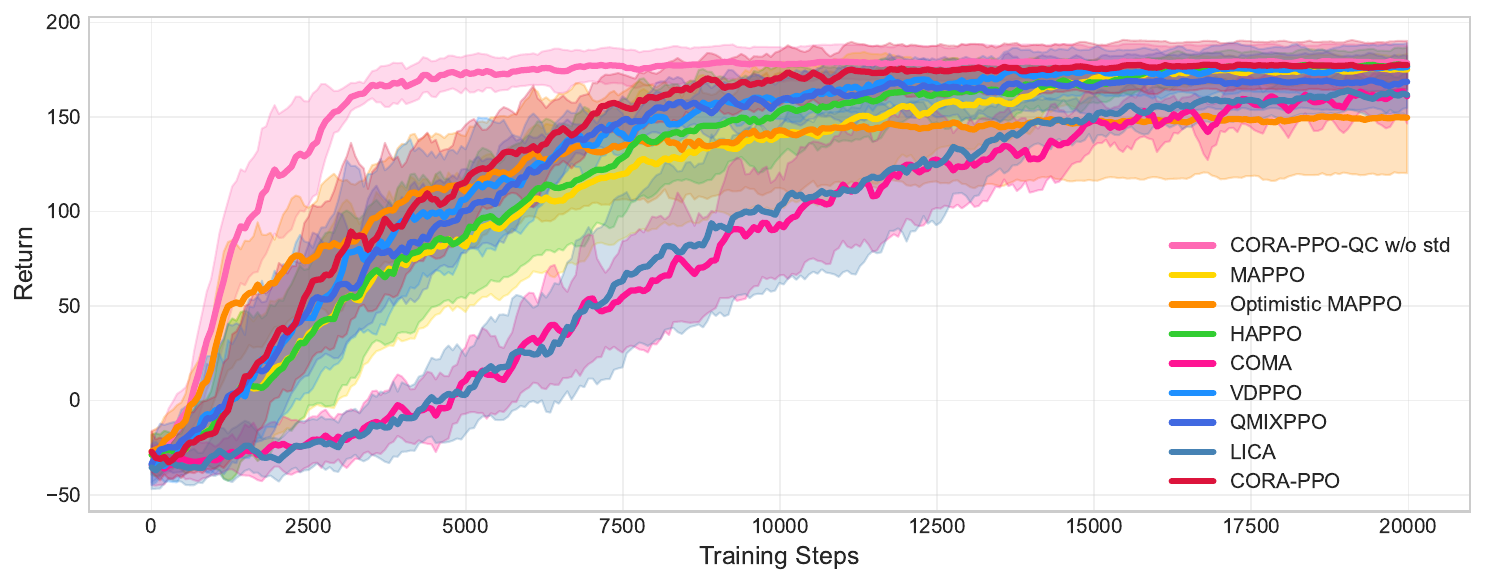}
			\label{fig:matrix_game_peaks_5}
		}
		\par\medskip
		\subfloat[10 Peaks (Matrix Team Game)]{%
			\includegraphics[width=0.48\columnwidth]{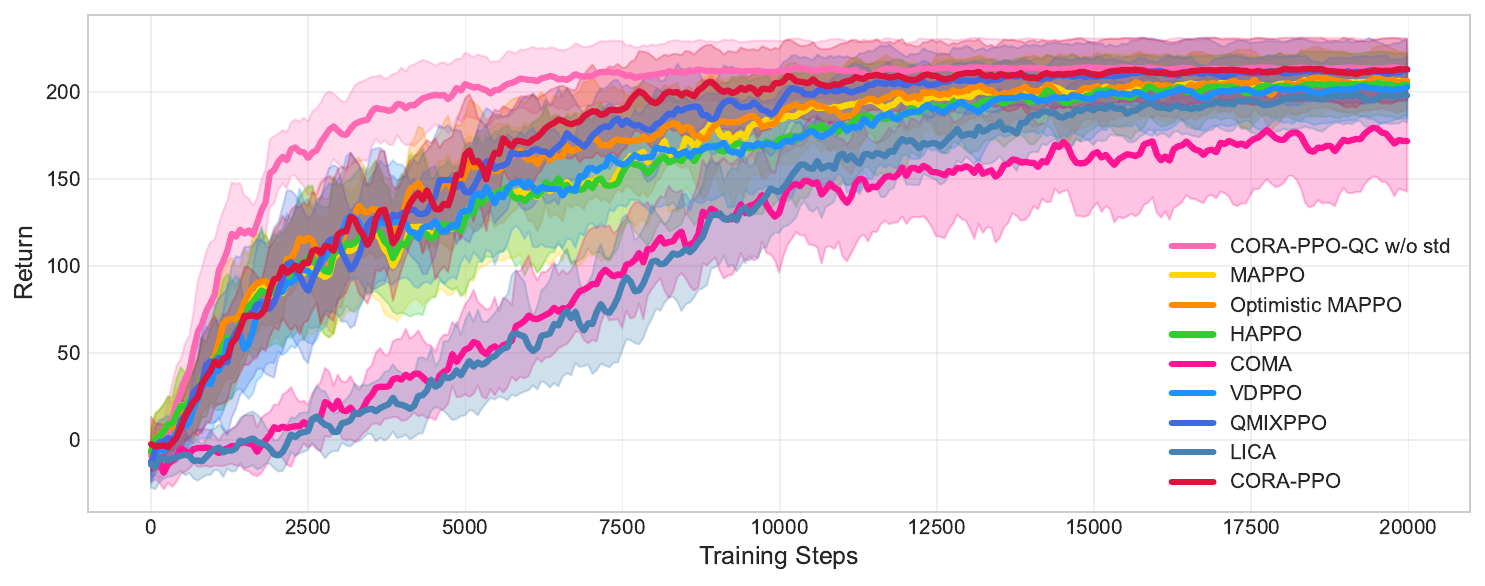}
			\label{fig:matrix_game_peaks_10}
		}
		\hfill
		\subfloat[15 Peaks (Matrix Team Game)]{%
			\includegraphics[width=0.48\columnwidth]{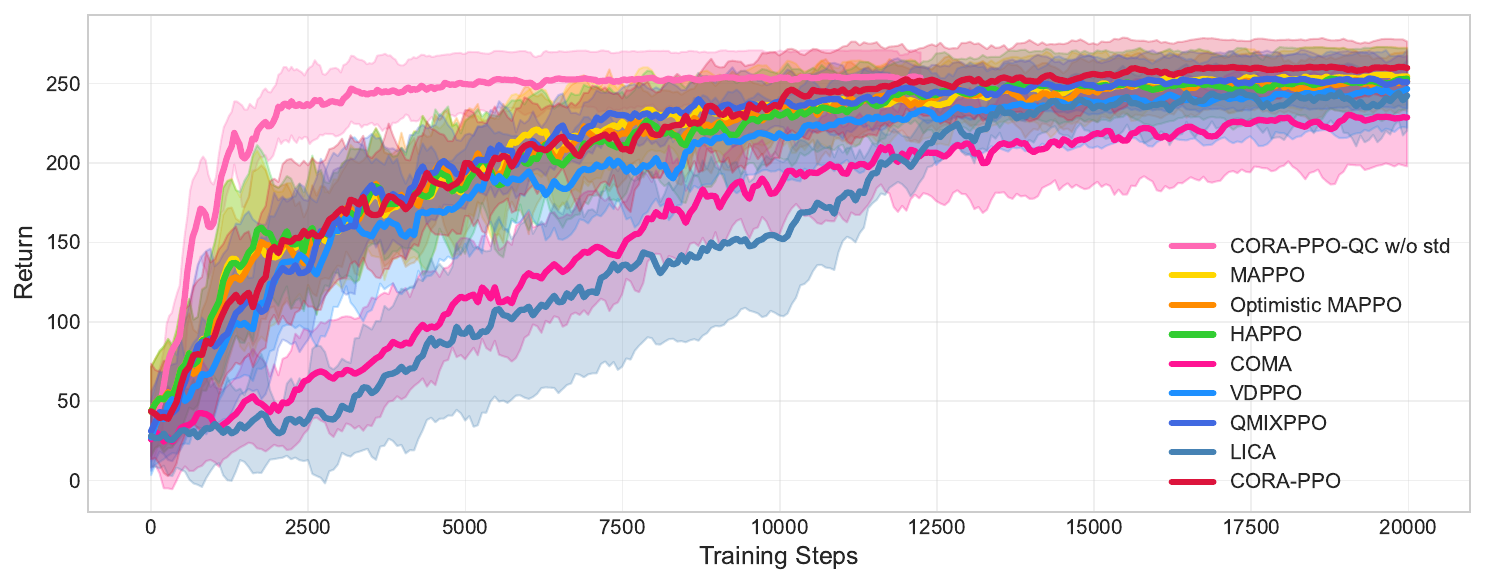}
			\label{fig:matrix_game_peaks_15}
		}
		\caption{
			Training performance on Matrix Team Game and its Multi-Peak variants with 5, 10, and 15 reward peaks.
		}
		\label{fig:matrix_team_game}
	\end{figure}
	
	As shown in Figure~\ref{fig:matrix_team_game}, CORA exhibits faster convergence and higher returns compared to the baseline algorithms, demonstrating superior coordination and learning efficiency in this general environment.
	
	Furthermore, inspired by recent methods in scalable multi-agent value factorization and coordination graphs \cite{christianos2023pareto, bohmer2020deep}, we implemented a quadratic critic network to efficiently approximate the joint action-value function. This architecture parameterizes $Q(s, a)$ as a state-dependent quadratic form:
	\[
		Q(s,a) = b(s) + \sum_{i \in N} \langle u_i(s), a_i \rangle + \sum_{i<j} a_i^\top W_{ij}(s) a_j,
	\]
	where $b(s)$ is a global state-dependent baseline, $u_i(s)$ represents the linear utility vector for agent $i$, and $W_{ij}(s)$ denotes the pairwise interaction matrix. The input $a_i$ corresponds to either a one-hot action vector or a policy probability distribution.

	While previous works primarily utilize this pairwise structure for message-passing and efficient joint-action maximization \cite{bohmer2020deep, christianos2023pareto}, the main advantage of this architecture in our framework is its ability to compute the previously defined coalitional value $Q_C(s,a_C)$ in closed form. Due to the network's linear and quadratic dependency on actions, the expectation over non-coalition agents ($j \notin C$) can be resolved analytically by simply replacing their action inputs $a_j$ with their current policy distributions $\pi_j(\cdot| o_j)$. This exact evaluation bypasses high-variance Monte Carlo sampling. The empirical results, denoted as CORA-PPO-QC, confirm that this precise estimation further highlights the stability advantage of CORA.
	
	\subsection{Differential Games}
	
	To demonstrate the learning process, we designed a 2D differential game environment similar to \cite{Ermo2016Lenient}. Each agent selects an action $x_1, x_2 \in [-5, 5]$ at every step. The reward function $R(x_1, x_2)$ is composed of a sum of several two-dimensional Gaussian potential fields, defined as:
	\begin{equation}
		R(x_1, x_2) = \sum_{i=1}^n h_i \cdot \exp\left( -\frac{(x_1 - c_{x_i})^2 + (x_2 - c_{y_i})^2}{\sigma_i^2} \right)
	\end{equation}
	Here, $n$ is the number of fields, $(c_{x_i}, c_{y_i})$ is the center of the $i$-th potential field, $h_i \in [5, 10]$ indicates the peak height of the potential field, and $\sigma_i \in [1, 2]$ controls its spread. This setup results in an environment with multiple local optima, presenting significant strategy exploration and learning challenges for MARL algorithms. The environment state itself does not evolve and can be regarded as a repeated single-step game. Key parameters like location, height, width of potential fields are set by a random seed.
	
	Figure~\ref{fig:training_curve_dg} compares the performance of MAPPO, HAPPO, CORA-PPO, CORA-PPO without the Std term (CORA-PPO w/o std), and Optimistic MAPPO in this environment. Among them, CORA-PPO achieves the best overall performance. Figure~\ref{fig:trajectories_dg} further illustrates the corresponding learning trajectories by showing how the policy means $\mu_i$ in the Gaussian policies $N(\mu_i, \sigma_i)$ evolve during training. Compared with the baselines, the CORA-PPO variants guide the agents more effectively toward optimal cooperative strategies, with trajectories converging to the peak regions of the 3D reward surface and the brightest regions of the 2D heatmap. In particular, the Std term leads to more stable trajectory convergence, whereas the trajectories of CORA-PPO w/o std are more dispersed.
	
	\begin{figure*}[!t]
		\centering
		\subfloat[MAPPO]{
			\includegraphics[width=0.31\textwidth]{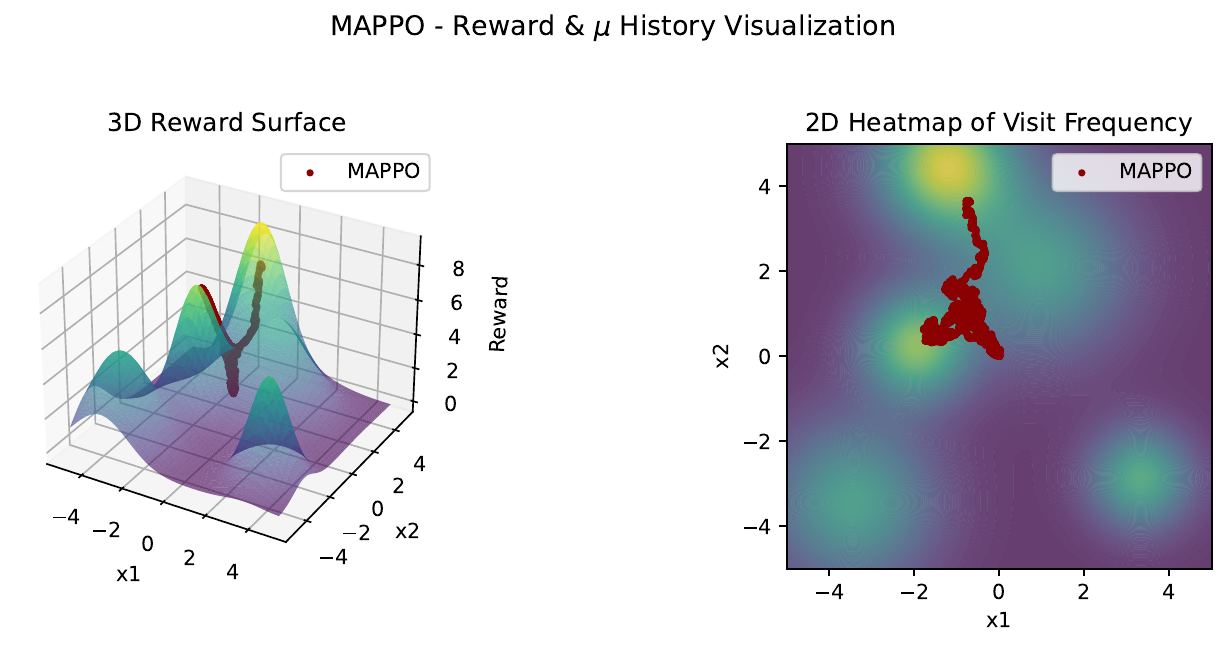}
			\label{fig:mappo_dg}
		}
		\subfloat[Optimistic MAPPO]{
			\includegraphics[width=0.31\textwidth]{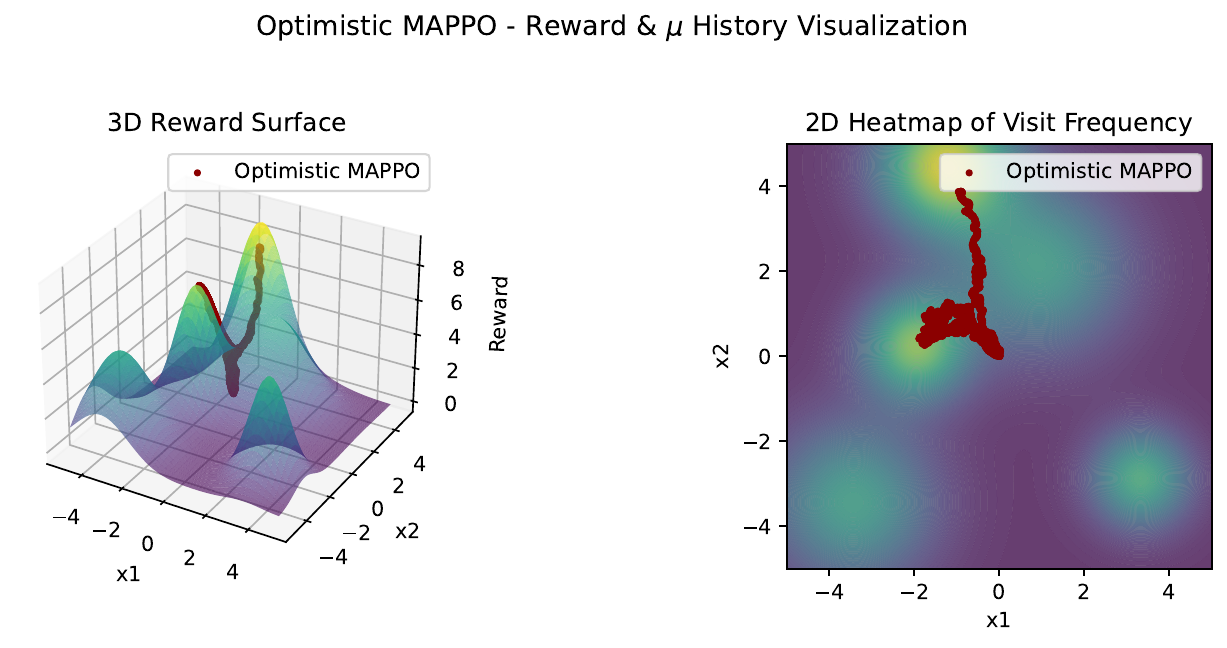}
			\label{fig:optimistic_mappo_dg}
		}
		\subfloat[HAPPO]{
			\includegraphics[width=0.31\textwidth]{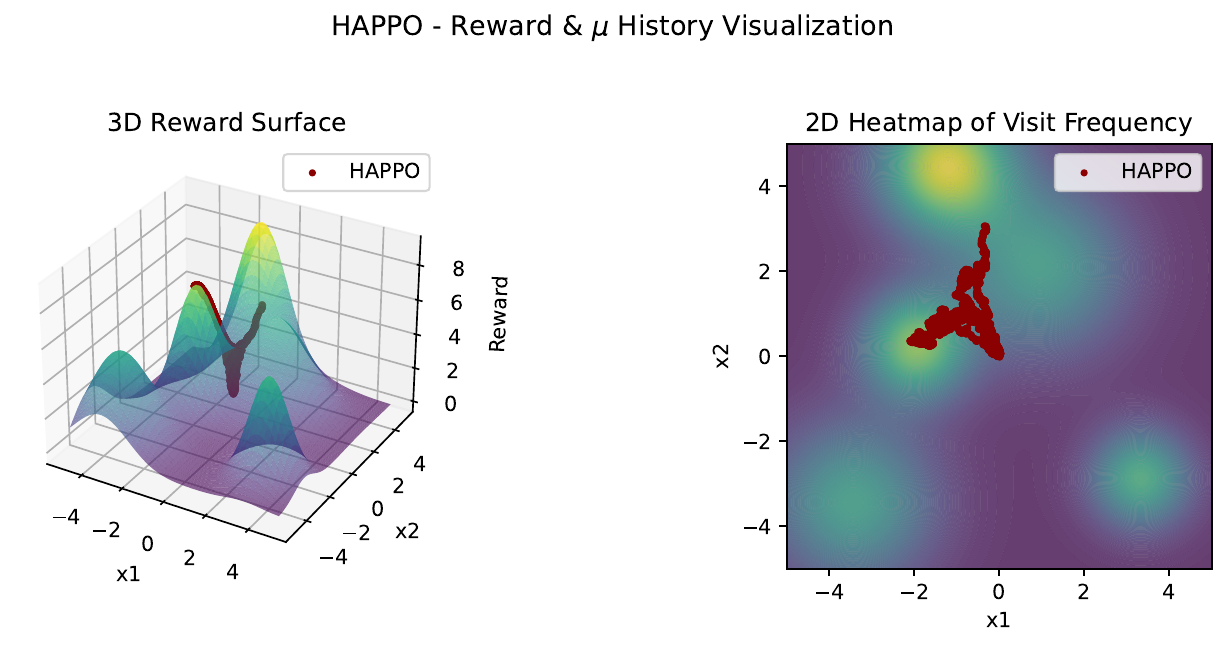}
			\label{fig:happo_dg}
		}
		\par\medskip
		
		\subfloat[CORA-PPO]{
			\includegraphics[width=0.31\textwidth]{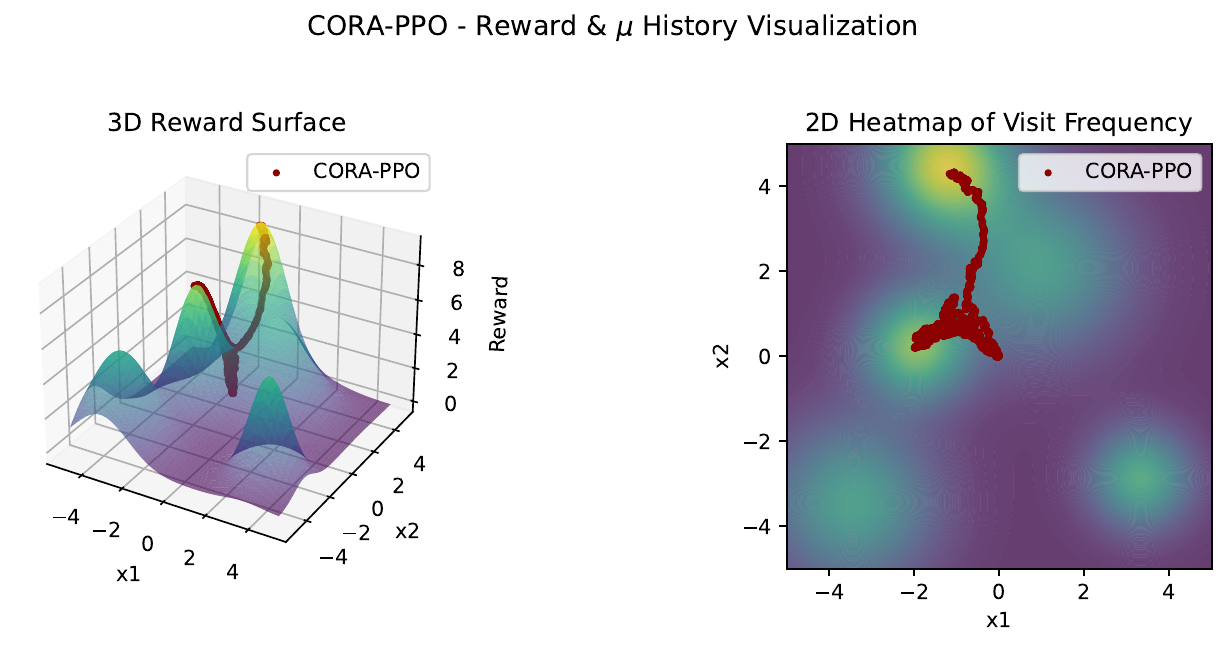}
			\label{fig:coreppo_dg}
		}
		\subfloat[CORA-PPO w/o std]{
			\includegraphics[width=0.31\textwidth]{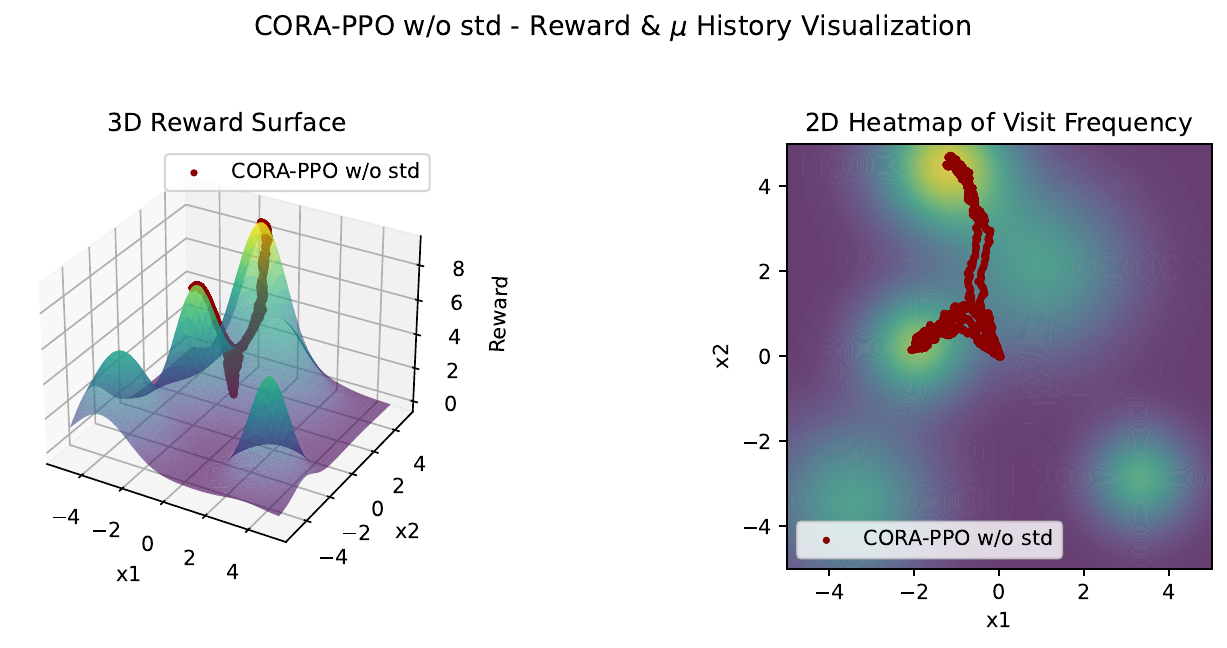}
			\label{fig:coreppo_nostd_dg}
		}
		\subfloat[All Methods]{
			\includegraphics[width=0.31\textwidth]{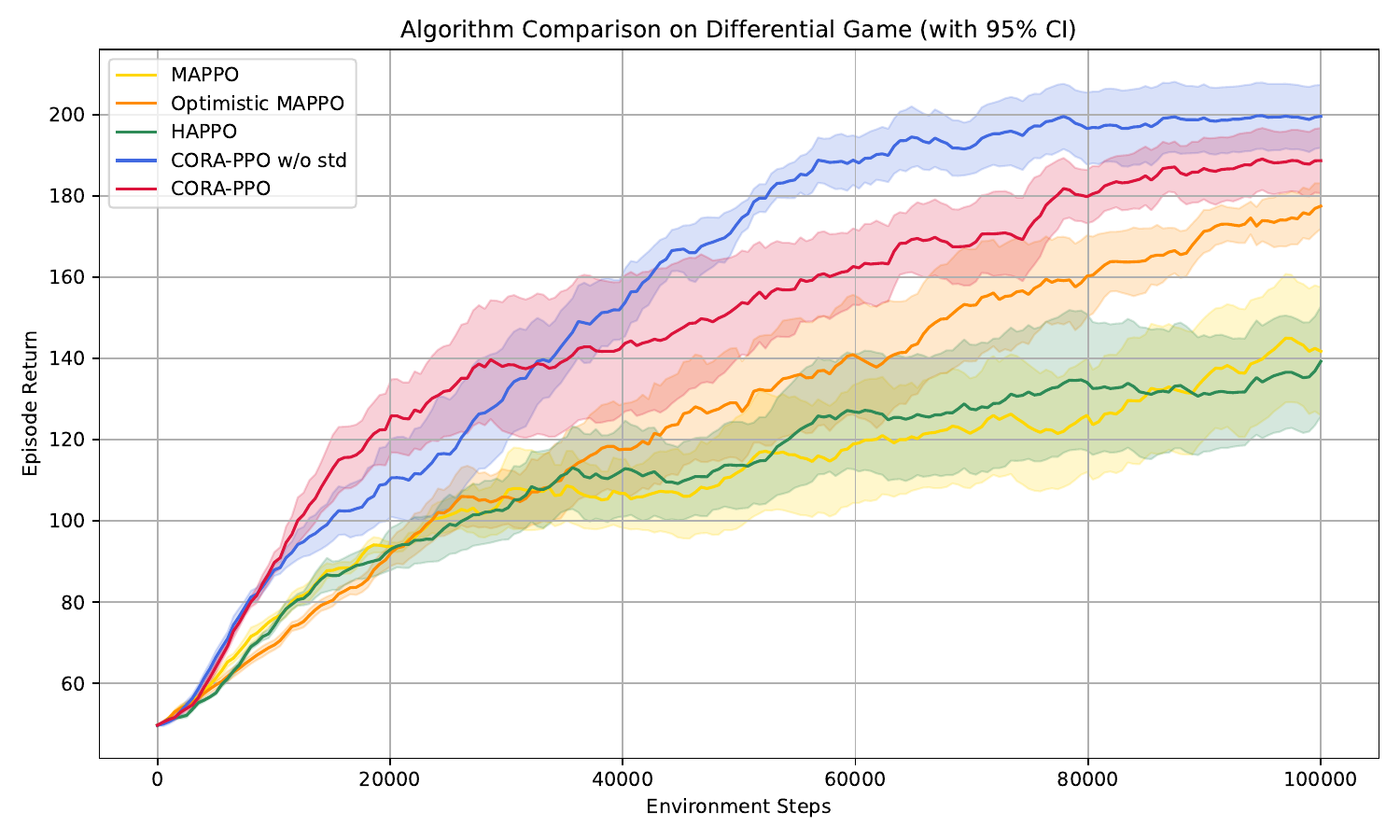}
			\label{fig:training_curve_dg}
		}
		\caption{The reward and learning trajectories of various algorithms in the differential game scenario ($\mu$ in Gaussian strategy).}
		\label{fig:trajectories_dg}
	\end{figure*}
	
	\subsection{VMAS}
	
	VMAS (Vectorized Multi-Agent Simulator) is a PyTorch-based vectorized multi-agent simulator designed for efficient multi-agent reinforcement learning benchmarking \cite{bettini2022vmas, bettini2024controlling, bou2023torchrl}. It provides a range of challenging multi-agent scenarios, and utilizes GPU acceleration, making it suitable for large-scale MARL training. We selected the following scenarios for testing: \textbf{Multi-Give-Way}: Four agents must coordinate to cross a shared corridor by giving way to each other to reach their respective goals. \textbf{Give-Way}: Two agents are placed in a narrow corridor with goals on opposite ends. Success requires one agent to yield and allow the other to pass first, reflecting asymmetric cooperative behavior. \textbf{Navigation}: Agents are randomly initialized and must navigate to their own goals while avoiding collisions.
	
	As shown in Figure \ref{fig:vmas}, CORA achieves higher returns and more stable performance compared to the other algorithms.
	
	\begin{figure*}[!t]
		\centering
		\subfloat[Give-Way]{
			\includegraphics[width=0.31\textwidth]{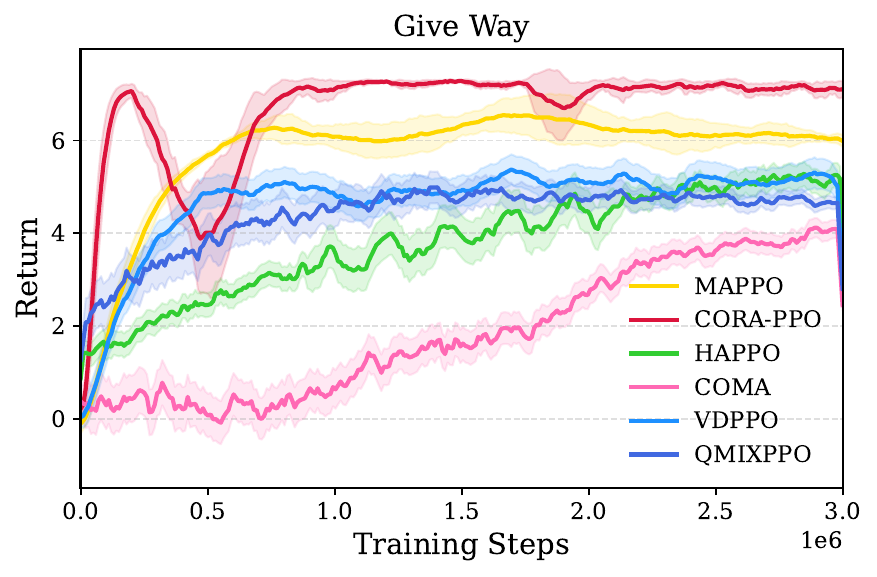}
		}
		\subfloat[Multi-Give-Way]{
			\includegraphics[width=0.31\textwidth]{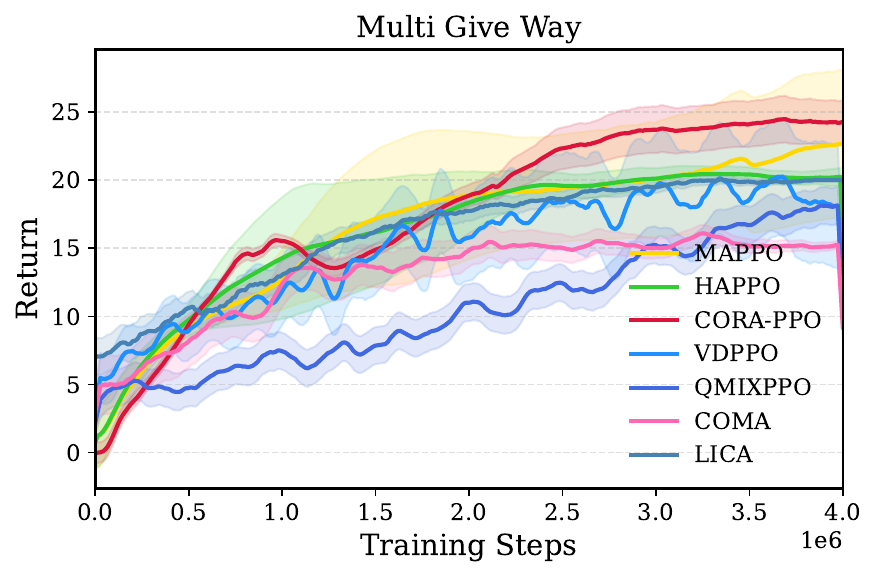}
		}
		\subfloat[Navigation]{
			\includegraphics[width=0.31\textwidth]{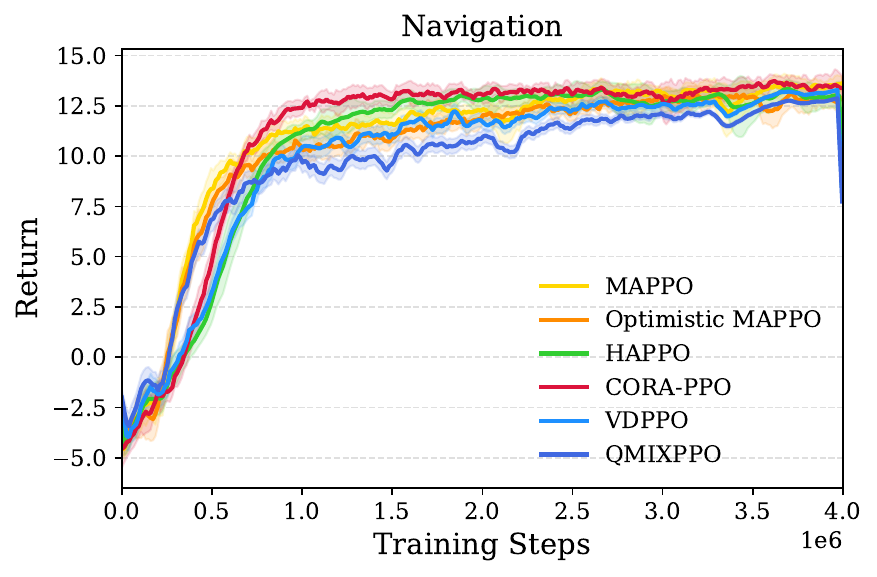}
		}
		\caption{Training performance on the VMAS scenarios.}
		\label{fig:vmas}
	\end{figure*}
	
	\subsection{Multi-Agent Mujoco}
	
	To demonstrate the effectiveness of CORA in continuous control tasks, we conducted experiments on the popular benchmark Multi-Agent MuJoCo (MA-MuJoCo) \cite{kuba2021trust}, using its latest version, MaMuJoCo-v5 \cite{gymnasium_robotics2023github}.
	
	As shown in Figure~\ref{fig:mamujoco}, except for the \textit{HalfCheetah 2x3} task where HAPPO slightly outperforms, CORA-PPO demonstrates superior results in the \textit{Ant 4x2}, \textit{HalfCheetah 6x1}, \textit{Walker2d 2x3}, and \textit{Hopper 3x1} tasks. These results highlight the effectiveness of CORA in handling diverse and challenging multi-agent continuous control environments.
	
	\begin{figure*}[!t]
		\centering
		\subfloat[Ant 2x4]{\includegraphics[width=0.31\textwidth]{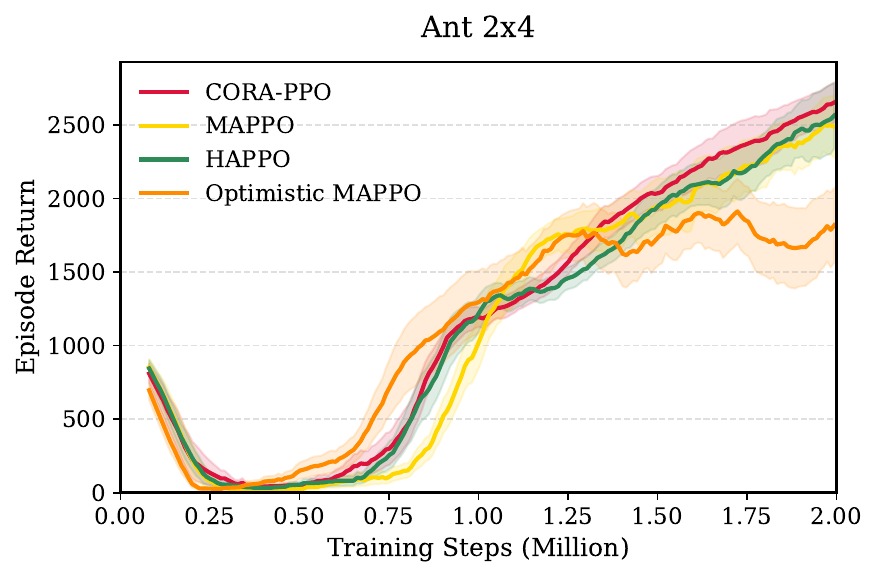}}
		\subfloat[Ant 4x2]{\includegraphics[width=0.31\textwidth]{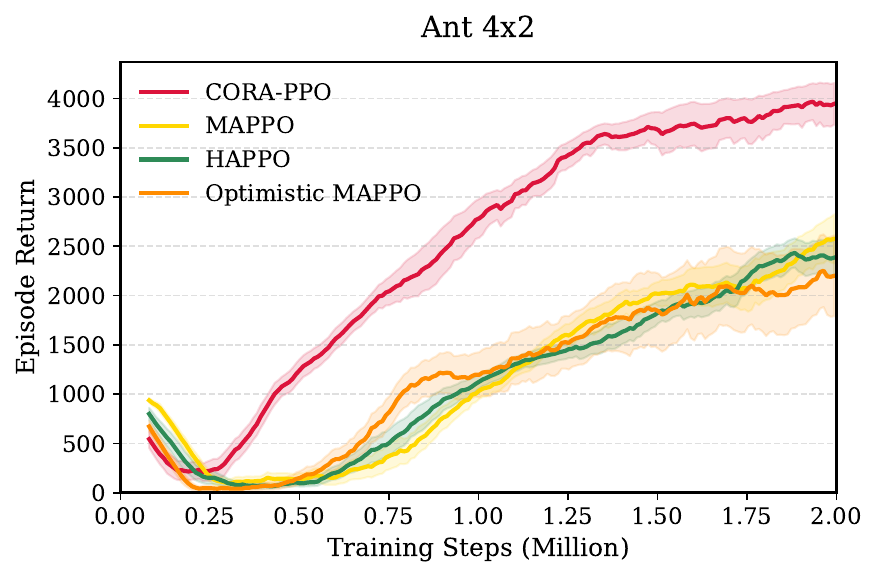}}
		\subfloat[HalfCheetah 2x3]{\includegraphics[width=0.31\textwidth]{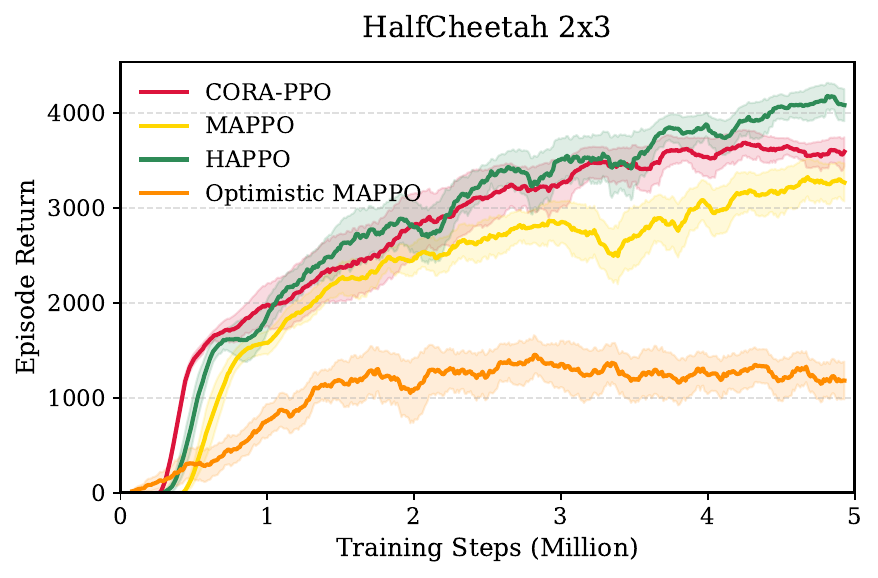}}
		\par\medskip
		
		\subfloat[HalfCheetah 6x1]{\includegraphics[width=0.31\textwidth]{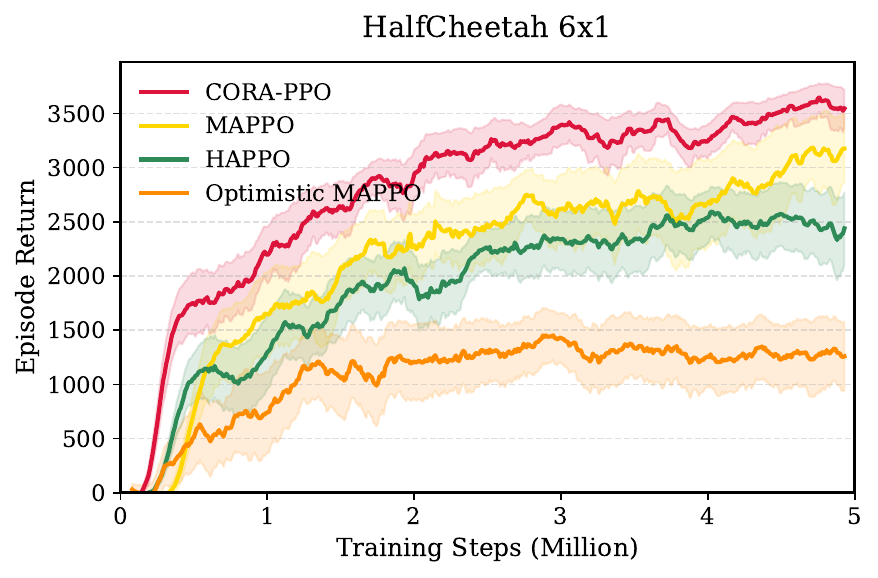}}
		\subfloat[Hopper 3x1]{\includegraphics[width=0.31\textwidth]{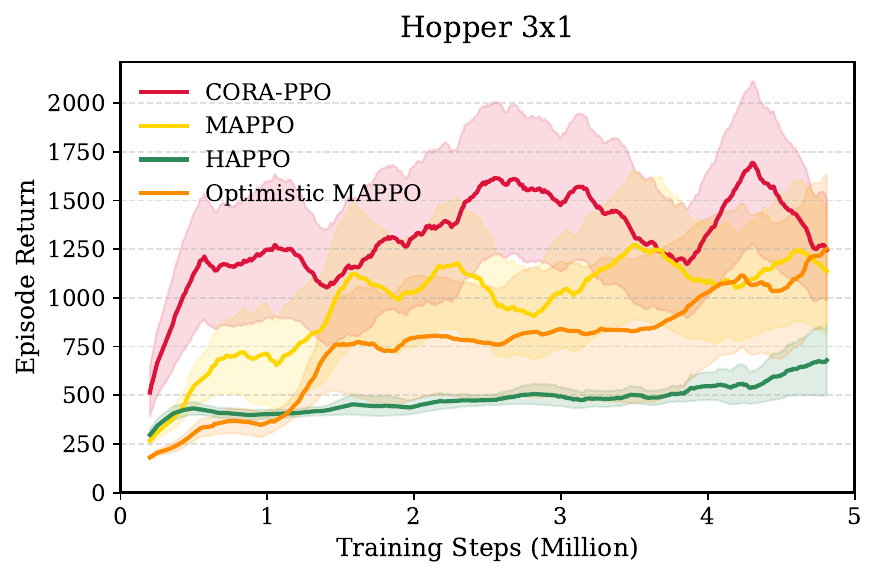}}
		\subfloat[Walker2d 2x3]{\includegraphics[width=0.31\textwidth]{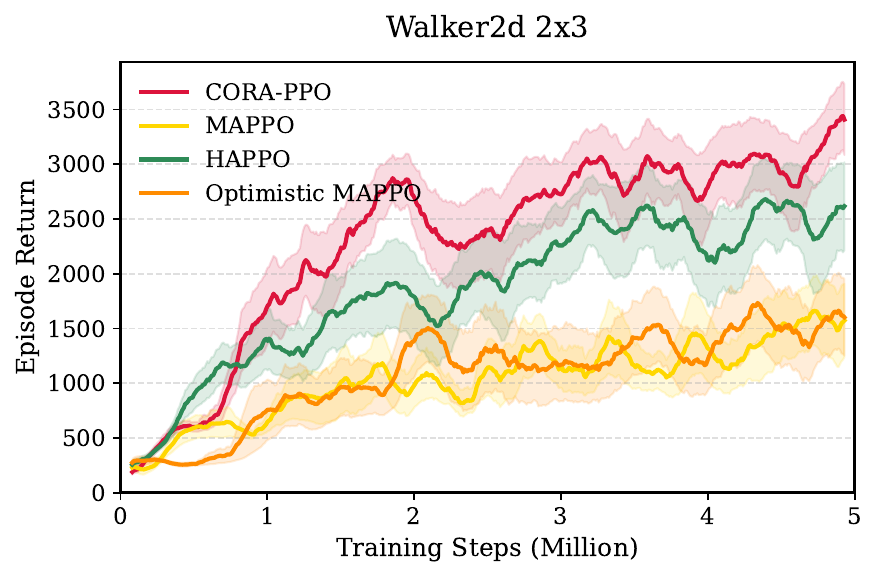}}
		\caption{Performance comparison on Multi-Agent MuJoCo (MaMuJoCo-v5) scenarios.}
		\label{fig:mamujoco}
	\end{figure*}
	
	\subsection{StarCraft Multi-Agent Challenge (SMAC)}
	
	We further evaluate the cooperative performance of CORA-PPO on the StarCraft Multi-Agent Challenge (SMAC). All experiments are conducted on SC2 version 4.10 under decentralized execution with a shared global team reward, following standard protocols. Results are averaged over 8 runs, with 95\% confidence intervals reported.
	
	We select five representative maps with different levels of cooperative difficulty: 3s\_vs\_5z, 8m, 2s\_vs\_1sc, 2s3z, and 3m. These scenarios vary in team composition, spatial complexity, and micro-management requirements, posing substantial challenges for credit assignment and coordinated tactical control.
	
	Figures~\ref{fig:smac_3s5z}--\ref{fig:smac_3m} show that CORA-PPO consistently attains higher win rates and faster convergence than MAPPO and HAPPO across all selected maps. On the more challenging scenarios, particularly \textit{3s\_vs\_5z} and \textit{2s\_vs\_1sc}, CORA-PPO exhibits stronger cooperative control and better asymptotic performance, highlighting its effectiveness under partial observability and intensive agent interactions.

	\begin{figure*}[!t]
		\centering
		\subfloat[3s\_vs\_5z]{\includegraphics[width=0.31\textwidth]{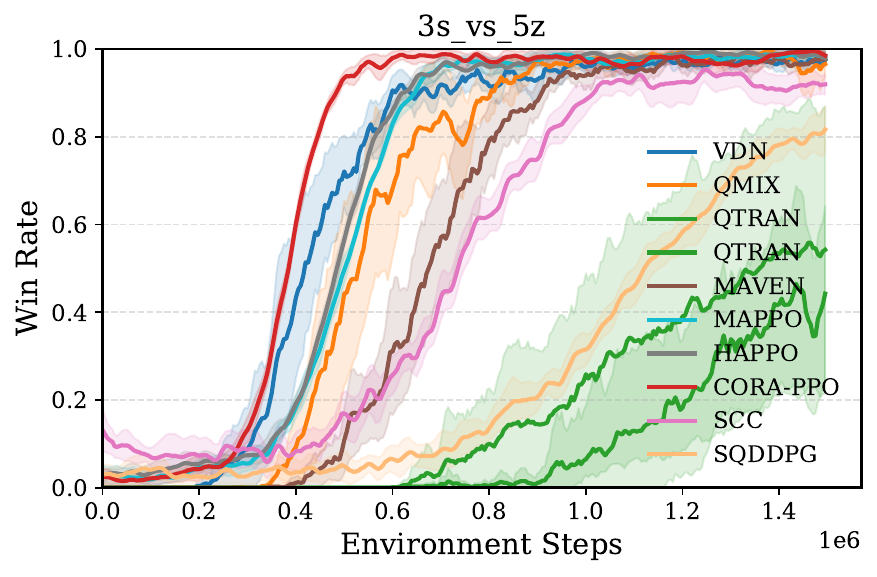}\label{fig:smac_3s5z}}
		\subfloat[8m]{\includegraphics[width=0.31\textwidth]{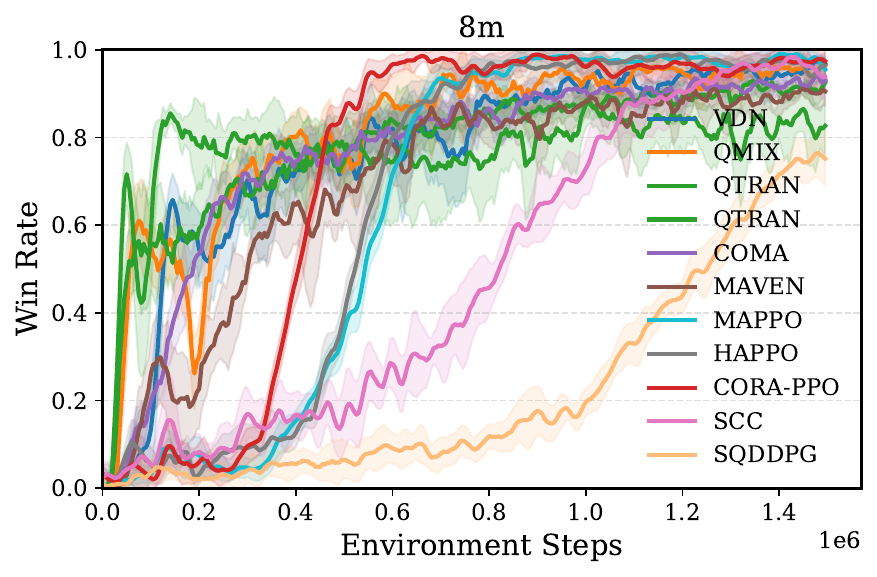}\label{fig:smac_8m}}
		\subfloat[2s\_vs\_1sc]{\includegraphics[width=0.31\textwidth]{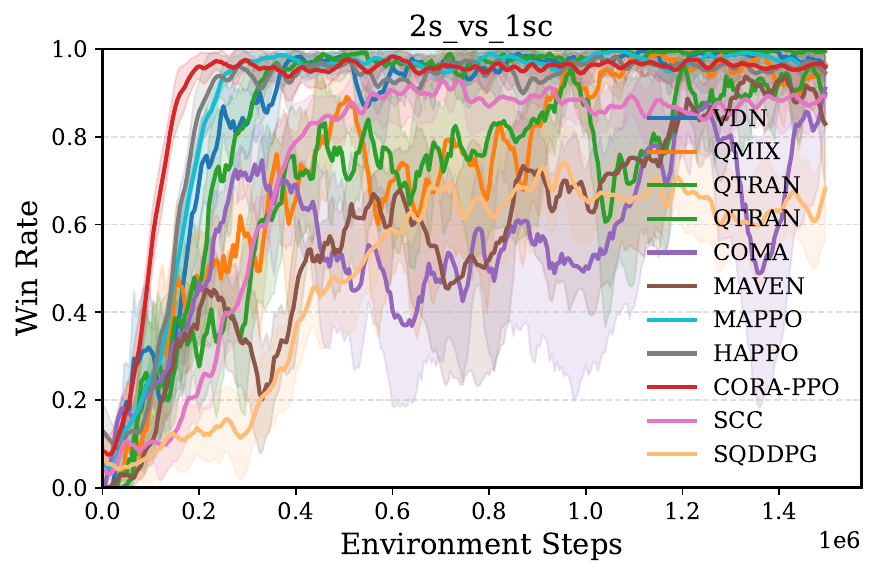}\label{fig:smac_2s1sc}}
		\par\medskip
		
		\subfloat[2s3z]{\includegraphics[width=0.31\textwidth]{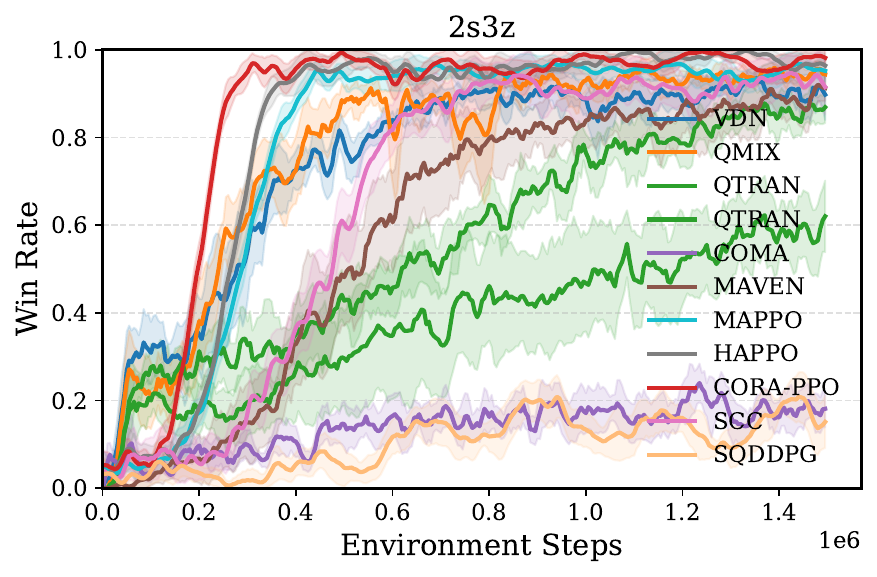}\label{fig:smac_2s3z}}
		\subfloat[3m]{\includegraphics[width=0.31\textwidth]{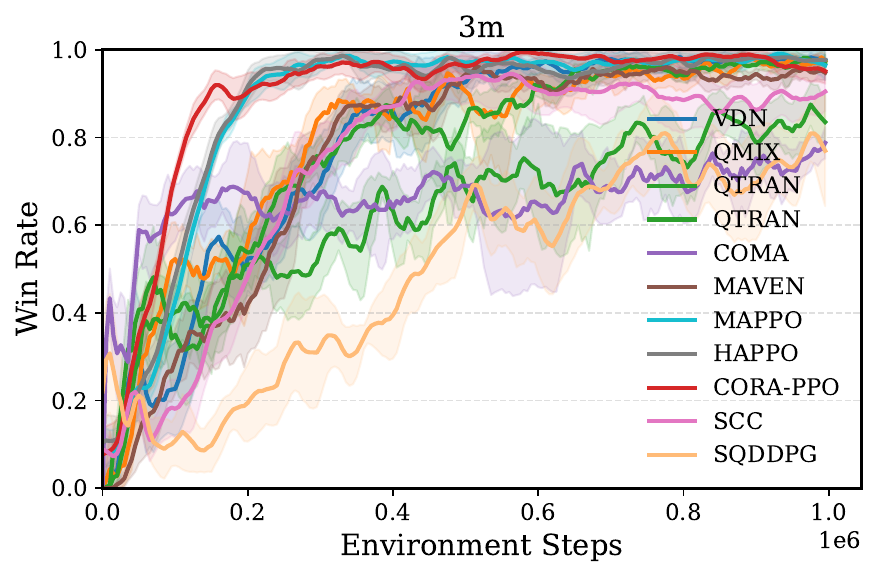}\label{fig:smac_3m}}
		\caption{Performance comparison on SMAC scenarios.}
		\label{fig:smac_results}
	\end{figure*}
	
	\subsection{Google Research Football}
	
	We further evaluate CORA-PPO on the Google Research Football (GRF) benchmark. Results are averaged over 8 random seeds, with 95\% confidence intervals reported. We consider three representative cooperative tasks: \textit{3 vs 1 with keeper} (3 agents), \textit{counterattack\_easy}, and \textit{counterattack\_hard}. These scenarios differ in offensive coordination requirements and planning difficulty, providing a meaningful test bed for cooperative decision-making under sparse and delayed rewards.

	Figures~\ref{fig:grf_3v1}--\ref{fig:grf_counterattack_hard} show that CORA-PPO consistently achieves higher returns and more stable training dynamics than the baselines.
	
	\begin{figure*}[!t]
		\centering
		\subfloat[3 vs 1 with keeper]{\includegraphics[width=0.31\textwidth]{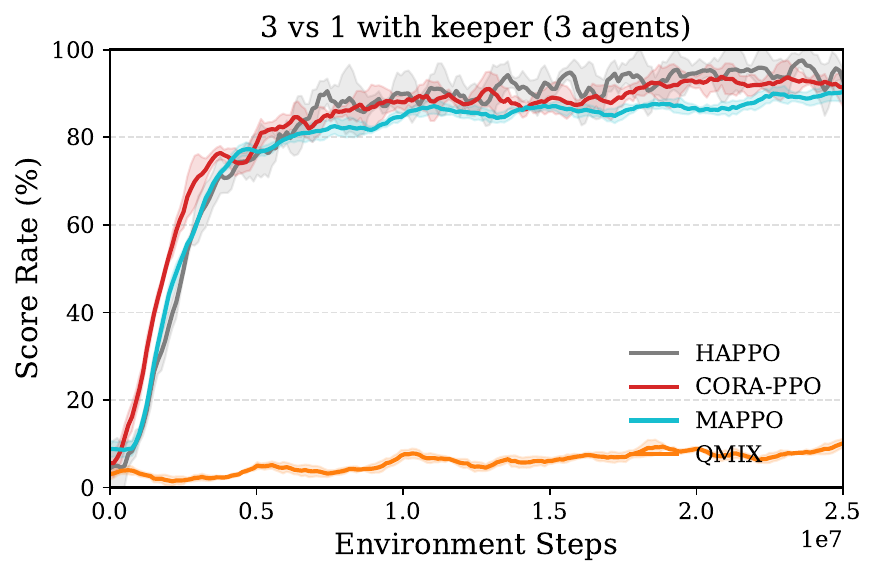}\label{fig:grf_3v1}}
		\subfloat[counterattack\_easy]{\includegraphics[width=0.31\textwidth]{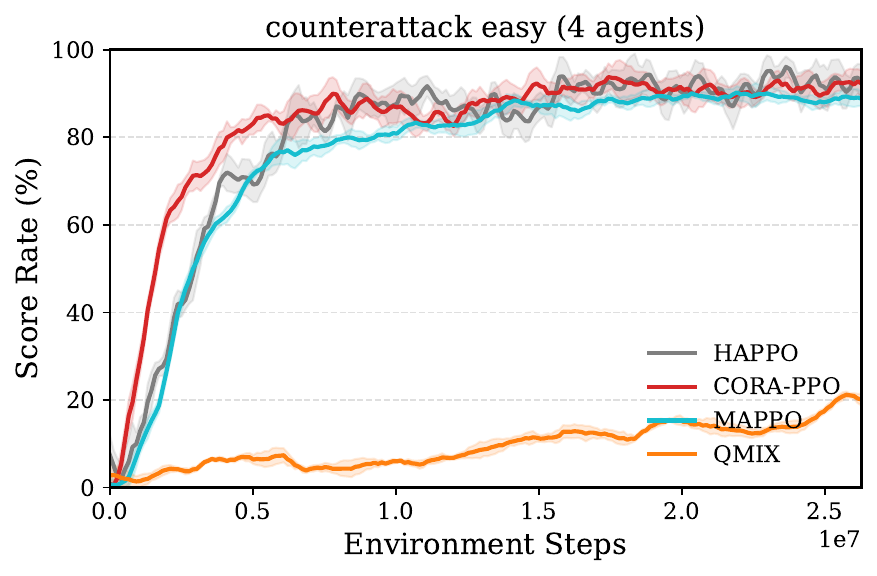}\label{fig:grf_counterattack_easy}}
		\subfloat[counterattack\_hard]{\includegraphics[width=0.31\textwidth]{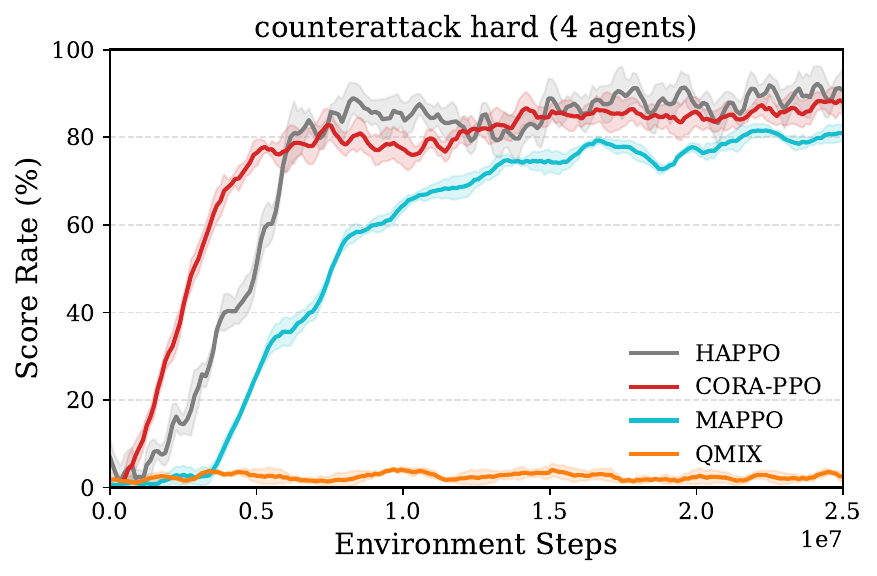}\label{fig:grf_counterattack_hard}}
		\caption{Performance comparison on GRF scenarios.}
		\label{fig:grf_results}
	\end{figure*}

	\subsection{Ablation Study of Coalition Sample Size and the Std Term}

	To investigate the effect of coalition sample size on performance, we conduct an ablation study in a differential game with 5 agents, where each agent has a one-dimensional action and the resulting joint action space is 5-dimensional. Considering the computational cost of larger-scale multi-agent settings, we focus on this setting because it provides a reasonable balance between task complexity and computational tractability.

	\begin{figure}[htbp]
		\centering
		\includegraphics[width=0.48\columnwidth]{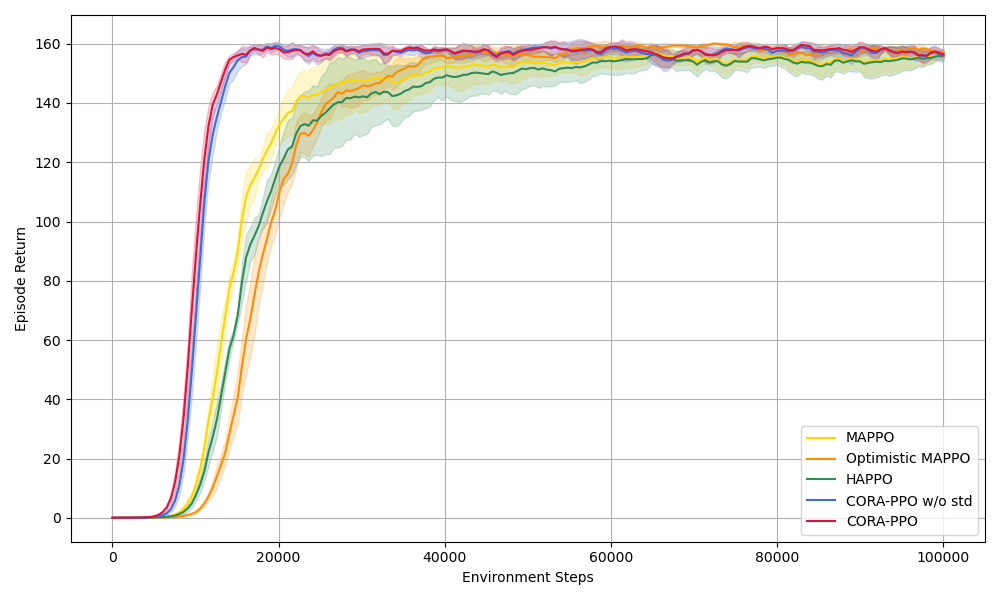}
		\hfill
		\includegraphics[width=0.48\columnwidth]{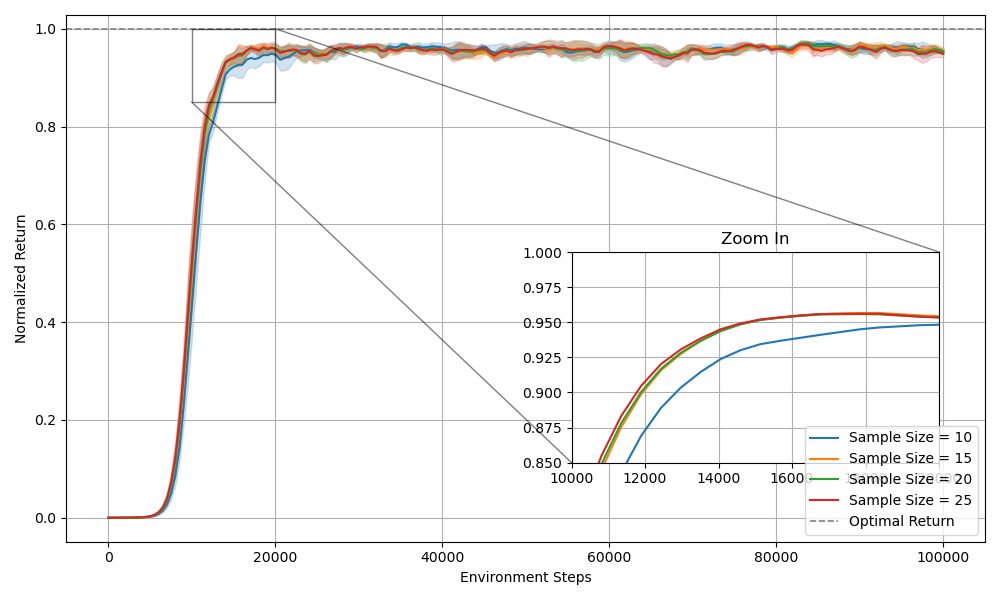}
		\caption{Training performance in 5-agent differential game. \textbf{Left}: Comparison among baseline methods. \textbf{Right}: Effect of coalition sampling size (sample sizes = 10, 15, 20, 25; full coalition size is $2^n - 2 = 30$). All algorithms are repeated 5 times to obtain a 95\% confidence interval. Key hyperparameters: Actor learning rate $5 \times 10^{-5}$, Critic learning rate $5 \times 10^{-4}$, $\gamma=0.99$, GAE $\lambda=0.95$, 10 epochs per update, clip $\epsilon_{\text{clip}}=0.2$, and 4 parallel environments.}
		\label{fig:sample_size_ablation}
	\end{figure}

	As shown in Figure~\ref{fig:sample_size_ablation}, increasing the coalition sample size generally improves performance, particularly in the early stages of training, as highlighted in the zoomed-in window. However, even with smaller sampling sizes (e.g., 10 or 15), the CORA still achieves competitive results. This indicates that CORA is robust to sample efficiency and remains effective under reduced computation, making it applicable to environments with a moderate number of agents. In addition, CORA with the Std term often improves performance.

	\begin{figure}[htbp]
		\centering
		\includegraphics[width=\columnwidth]{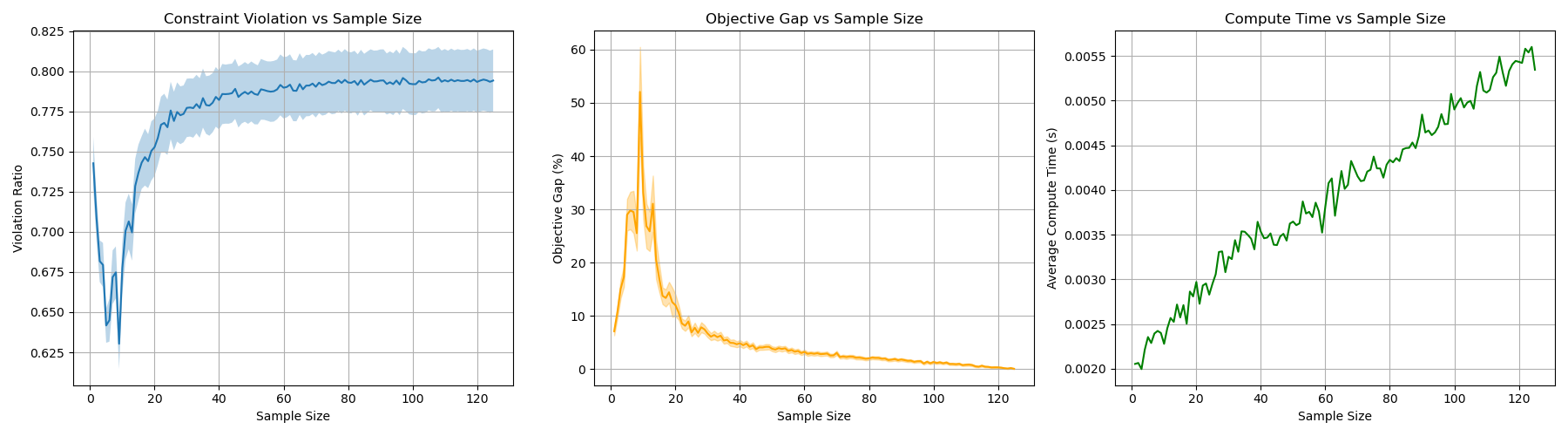}
		\caption{Approximation error and computational cost of sampled credit assignment. Violation Ratio denotes the proportion of violated coalition-rationality constraints; Objective Gap denotes the relative difference in objective value compared with the full solution; Compute Time denotes the average runtime across trials. Results are reported for 7 agents, with randomly generated advantage functions over 20 trials.}
		\label{fig:approx_credit}
	\end{figure}

	Figure~\ref{fig:approx_credit} shows that a relatively small number of sampled coalitions already provides an accurate and computationally efficient approximation. Although the violation ratio increases slightly when fewer coalitions are sampled, the objective gap remains small while computation time is substantially reduced. These results support sampled credit assignment as a practical alternative in larger-scale settings, where enumerating all coalition constraints becomes computationally prohibitive.

	\section{Conclusion}
	\label{sec:conclusion}

	This paper studies credit assignment in cooperative multi-agent policy gradient methods from a coalitional perspective. We proposed CORA, a core-based advantage allocation framework that evaluates coalition-wise advantages and computes per-agent credits through a regularized least $\epsilon$-core formulation. In this way, CORA preserves global consistency while ensuring that strategically valuable coalitions receive sufficient incentives during policy updates. We further established coalition-level policy-improvement bounds and a sampled-coalition approximation guarantee, which together clarify why CORA protects beneficial coalition behaviors and remains tractable in practice. Experiments on matrix games, differential games, VMAS, Multi-Agent MuJoCo, SMAC, and Google Research Football showed consistent improvements in learning efficiency, stability, and final performance over strong baselines. 
	
	Future work will consider more scalable coalition evaluation, other cooperative-game solution concepts for credit assignment, such as the CIS value, the Banzhaf value, and the solidarity value, and broader extensions to large-scale, partially observable, and heterogeneous multi-agent systems.
	
	\bibliography{reference}
	\bibliographystyle{IEEEtran}

	\begin{IEEEbiography}[{\includegraphics[width=1in,height=1.25in,clip,keepaspectratio]{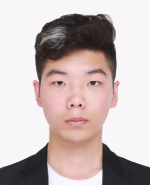}}]{Mengda Ji}
		is currently a Ph.D. student in the Unmanned System Research Institute, and the Key Laboratory of Intelligence, Games and Information Processing of Shaanxi Province, Northwestern Polytechnical University, China. He received his B.S. degree from Northwestern Polytechnical University, China, in 2021. His research interests include game theory and multi-agent reinforcement learning.
	\end{IEEEbiography}
	
	\begin{IEEEbiography}[{\includegraphics[width=1in,height=1.25in,clip,keepaspectratio]{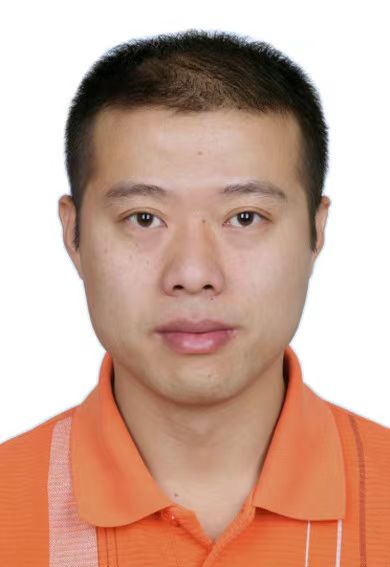}}]{Genjiu Xu}
		received the Ph.D. degree from University of Twente, Enschede, The Netherlands, in 2008. He is a Full Professor and Associate Dean with the School of mathematics and statistics, and the Leader of Key Laboratory of Intelligence, Games and Information Processing of Shaanxi Province, Northwestern Polytechnical University, China. He is currently the Vice Chairman of Game Theory Committee of China Operations Research Society. His research interests include game theory and operations research.
	\end{IEEEbiography}
	
	\begin{IEEEbiography}[{\includegraphics[width=1in,height=1.25in,clip,keepaspectratio]{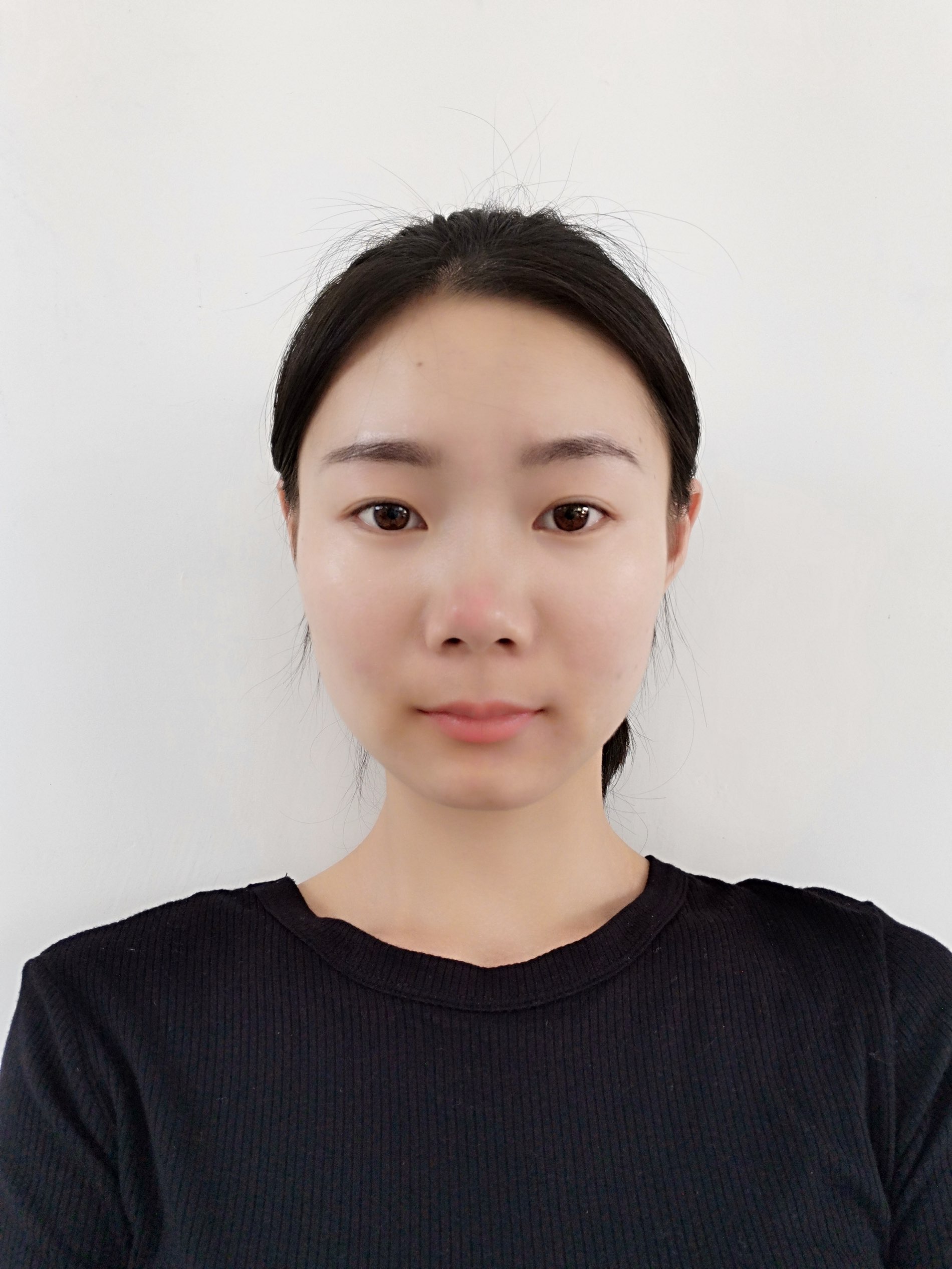}}]{Keke Jia}
		is currently a Ph.D. student in the School of mathematics and statistics, and the Key Laboratory of Intelligence, Games and Information Processing of Shaanxi Province, Northwestern Polytechnical University, China. Her research interests include game theory and collaboration of multi-agent system.
	\end{IEEEbiography}

	\begin{IEEEbiography}[{\includegraphics[width=1in,height=1.25in,clip,keepaspectratio]{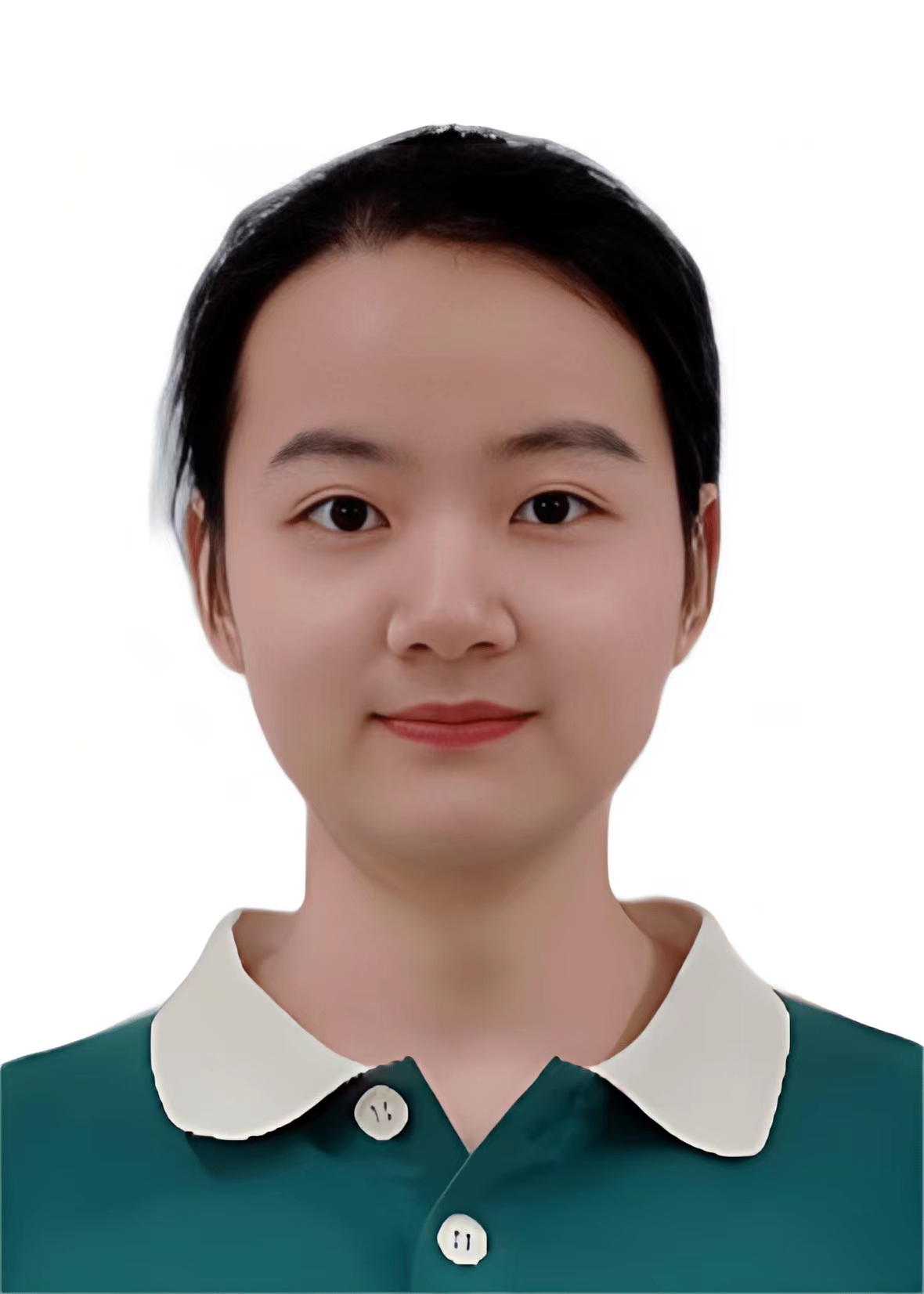}}]{Zekun Duan}
		received the B.S. degree in applied statistics from North China University of Water Resources and Electric Power, Zhengzhou, China, in 2019. She is currently pursuing the Ph.D. degree with the School of Mathematics and Statistics, Northwestern Polytechnical University, Xi’an, China. Her research interests include noncooperative game theory, multi-agent systems, and equilibrium computation.
	\end{IEEEbiography}
	
	\begin{IEEEbiography}[{\includegraphics[width=1in,height=1.25in,clip,keepaspectratio]{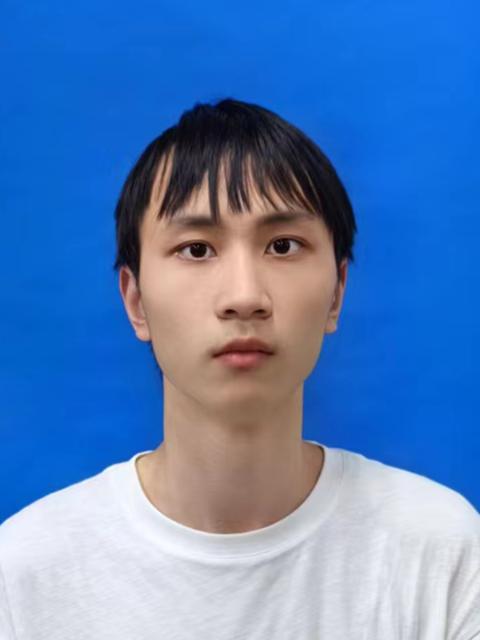}}]{Yong Qiu} is currently a postgraduate student at the School of Mathematics and Statistics, Northwestern Polytechnical University, Xi’an, China. He received his B.S. degree from North China Electric Power University, China, in 2025. His research interests include game theory, multi-agent reinforcement learning, and operations optimization.
	\end{IEEEbiography}
	
	\begin{IEEEbiography}[{\includegraphics[width=1in,height=1.25in,clip,keepaspectratio]{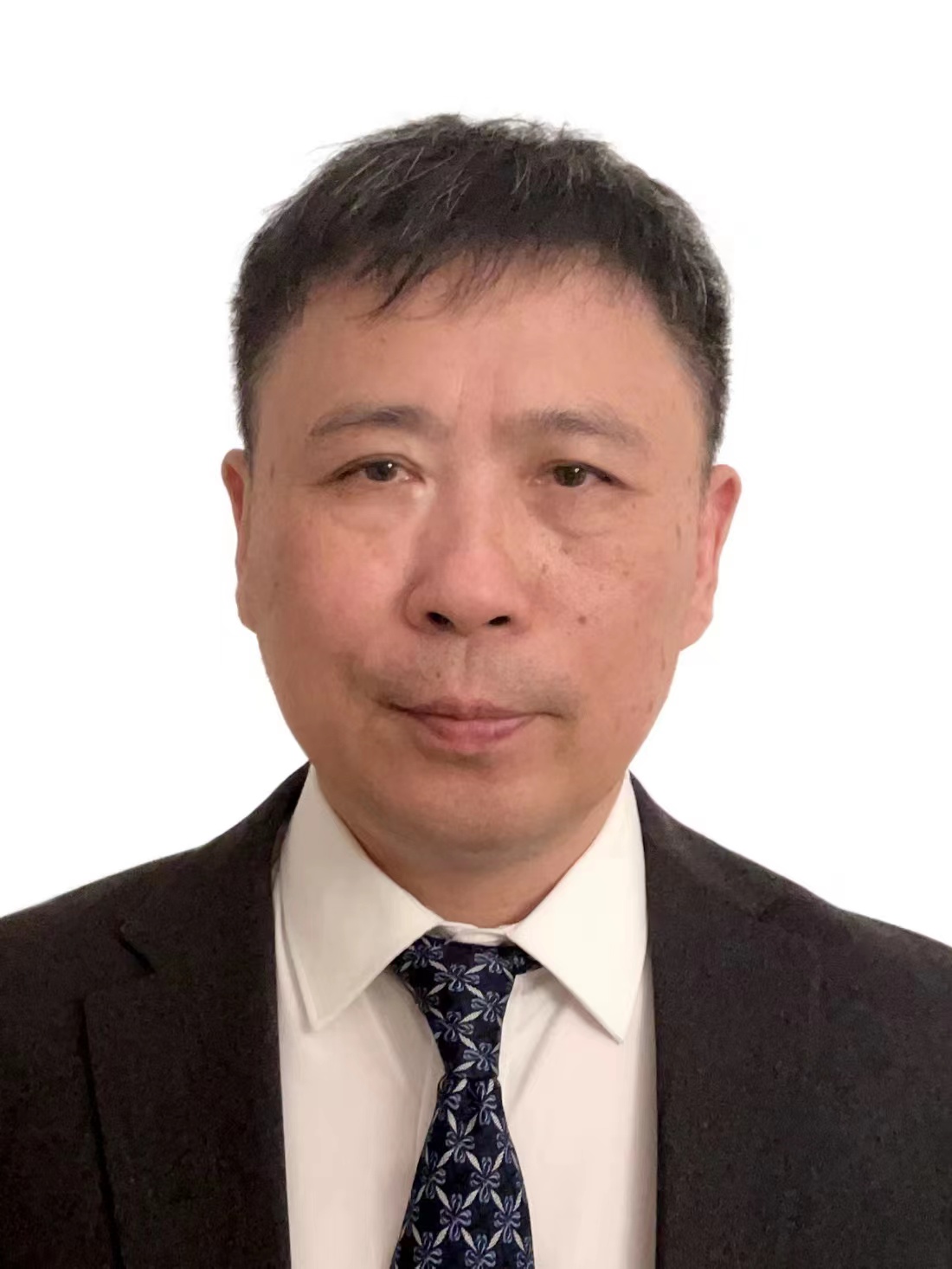}}]{Jianjun Ge}
		received the Ph.D. degree from Beijing Institute of Technology, China, in 2016. He is a Chief Scientist and Director with the Information Science Academy, China Electronics Technology Group Corporation. His research interests include information theory, radar systems and artificial intelligence.
	\end{IEEEbiography}
	
	\begin{IEEEbiography}[{\includegraphics[width=1in,height=1.25in,clip,keepaspectratio]{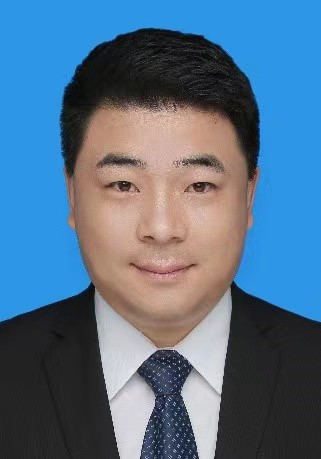}}]{Mingqiang Li}
		received the Ph.D. degree from the University of Chinese Academy of Sciences in 2017. He is a Senior Researcher with the Information Science Academy, China Electronics Technology Group Corporation. His main research interests include representation learning, optimization theory and intelligent decision-making.
	\end{IEEEbiography}
	\vfill
	
\end{document}